\documentclass[11pt,english]{article}
\usepackage[T1]{fontenc}
\usepackage[latin9]{inputenc}
\usepackage{color}
\definecolor{note_fontcolor}{rgb}{0.39453125, 0.39453125, 0.39453125}
\usepackage{babel}
\usepackage{verbatim}
\usepackage{float}
\usepackage{amsthm}
\usepackage{amsmath}
\usepackage{graphicx}
\usepackage[unicode=true,pdfusetitle,
 bookmarks=true,bookmarksnumbered=true,bookmarksopen=false,     bookmarksopenlevel=1,pdfpagemode=UseOutlines,
 breaklinks=false,pdfborder={0 0 1},backref=page,colorlinks=true]
 {hyperref}
\hypersetup{
 citecolor=darkgreen, linkcolor=darkblue}
\usepackage{breakurl}

\makeatletter

\newenvironment{lyxgreyedout}
  {\textcolor{note_fontcolor}\bgroup\ignorespaces}
  {\ignorespacesafterend\egroup}
\floatstyle{ruled}
\newfloat{algorithm}{tbp}{loa}
\providecommand{\algorithmname}{Algorithm}
\floatname{algorithm}{\protect\algorithmname}

\theoremstyle{plain}
\newtheorem{thm}{\protect\theoremname}
  \theoremstyle{definition}
  \newtheorem{defn}{\protect\definitionname}
  \theoremstyle{plain}
  \newtheorem{lem}{\protect\lemmaname}
  \theoremstyle{remark}
  \newtheorem{rem}{\protect\remarkname}
  \theoremstyle{plain}
  \newtheorem{cor}{\protect\corollaryname}

\usepackage{fullpage}
\usepackage{color} 
\usepackage[titletoc,title]{appendix}
\definecolor{darkgreen}{rgb}{0,0.4,0}
\definecolor{darkblue}{rgb}{0,0,0.6}
\definecolor{Yellow}{rgb}{0.2,0.2,0.2}

\makeatother

  \providecommand{\definitionname}{Definition}
  \providecommand{\lemmaname}{Lemma}
  \providecommand{\remarkname}{Remark}
\providecommand{\corollaryname}{Corollary}
\providecommand{\theoremname}{Theorem}

\begin{document}

\title{Dynamic Complexity of Planar 3-connected Graph Isomorphism}

\author{Jenish C. Mehta%
\thanks{This work was done while the author was interning at Chennai Mathematical
Institute in May-June, 2010%
}}
\maketitle
\begin{abstract}
Dynamic Complexity (as introduced by Patnaik and Immerman \cite{patnaik1994dyn})
tries to express how hard it is to \emph{update} the solution to a
problem when the input is changed \emph{slightly}. It considers the
\emph{changes} required to some stored data structure (possibly a
massive database) as small quantities of data (or a tuple) are inserted
or deleted from the database (or a structure over some vocabulary).
The main difference from previous notions of dynamic complexity is
that instead of treating the update quantitatively by finding the
the time/space trade-offs, it tries to consider the update \emph{qualitatively},
by finding the \emph{complexity class }in which the update can be
expressed (or made). In this setting, $\textsf{\textsf{DynFO}}$,
or Dynamic First-Order, is one of the smallest and the most natural
complexity class (since SQL queries can be expressed in First-Order
Logic), and contains those problems whose solutions (or the stored
data structure from which the solution can be found) can be updated
in First-Order Logic when the data structure undergoes small changes.

Etessami \cite{etessami1998dynamic} considered the problem of isomorphism
in the dynamic setting, and showed that Tree Isomorphism can be decided
in $\textsf{DynFO}$. In this work, we show that isomorphism of Planar
3-connected graphs can be decided in $\textsf{DynFO}^{+}$ (which
is $\textsf{\textsf{\textsf{\textsf{DynFO}}}}$ with some polynomial
precomputation). We maintain a canonical description of 3-connected
Planar graphs by maintaining a database which is accessed and modified
by First-Order queries when edges are added to or deleted from the
graph. We specifically exploit the ideas of Breadth-First Search and
Canonical Breadth-First Search to prove the results. We also introduce
a novel method for canonizing a 3-connected planar graph in First-Order
Logic from Canonical Breadth-First Search Trees.
\end{abstract}

\section{Introduction}

Consider the problem $\textsc{lis(A)}$ of finding the longest increasing
subsequence of a sequence (or array) of $n$ numbers $\textsc{A}$.
The ``template'' dynamic programming polynomial time solution proceeds
by subsequently finding and storing $\textsc{lis(A[1:\ensuremath{i}])}$
- the longest increasing subsequence of numbers from 1 to $i$ that
necessarily ends with the $i$'th number. $\textsc{lis(A[1:\ensuremath{i+1}])}$
is found, given $\textsc{lis(A[1:1])}$ to $\textsc{lis(A[1:\ensuremath{i}])}$,
by simply finding the maximum sequence formed by possibly appending
$\textsc{A[\ensuremath{i+1}]}$ to the largest subsequence from $\textsc{lis(A[1:1])}$
to $\textsc{lis(A[1:\ensuremath{i}])}$. 

This \emph{paradigm} of dynamic programming (or incremental thinking)\emph{,
}of \emph{storing} information using polynomial space, and \emph{updating}
it to get the required results, is neatly captured in the Dynamic
Complexity framework introduced by Patnaik and Immerman \cite{patnaik1994dyn}.
Broadly, Dynamic Complexity tries to measure or express how hard it
is to \emph{update} some stored information, so that some required
query can be answered. For instance, for some graph problem, like
reachability, it tries to measure (or express) how hard it is to update
some stored information when an edge is inserted or deleted from the
graph, so that the required query, like reachability between two vertices
$s$ and $t$, can be answered easily from the stored information.
Essentially, it asks how hard is one step of induction, or how hard
it is to update one step of some recurrence. 

This Dynamic Complexity framework (as in \cite{patnaik1994dyn}) differs
from other notions in two ways. For some problem (say a graph theoretic
problem like colorability or reachability), the traditional notions
of the dynamic complexity try to measure the amount of \emph{time
}and \emph{space }required to make some update to the problem (like
inserting/deleting edges from a graph or inserting/deleting tuples
from a database), and the trade-offs between the two. Dynamic Complexity,
instead tries to measure (or express) the resources required for an
update \emph{qualitatively}. Hence, it tries to measure an update
by the \emph{complexity class} in which it lies, rather than the explicit
time/space requirements. For any static complexity class $\mbox{\textsf{C}}$,
informally, the dynamic complexity class $\textsf{DynC}$ consists
of the set of problems, such that any \emph{update} (to be defined
formally later) to the problem can be expressed in the compexity class
$\textsf{C}$. A bit more formally, a language $L$ is in the dynamic
complexity class $\textsf{DynC}$ if we can maintain a tuple of relations
(say $T$) for deciding the language in $\textsf{C}$, such that after
any insertion or deletion of a tuple to the relations, they can be
effectively updated in the complexity class $\textsf{C}$ (updation
is required so that even after the insertion/deletion of the tuple,
they decide the same language \textbf{$L$}).

Another difference is that it treats the complexity classes in a Descriptive
manner (using the language of Finite Model Theory) rather than the
standard Turing manner (defined by tapes and movement of pointers).
Since Descriptive Complexity tries to measure the hardness of \emph{expressing}
a problem rather than the hardness of \emph{finding} a solution to
the problem, Dynamic Complexity tries to measure how hard it is to
\emph{express} an update to some problem. Though, since either definition
- Descriptive or Turing - lead to complexity classes with the same
expressive power, any of the definitions remain valid. 

Consider the dynamic complexity class $\textsf{DynP}$ (or $\textsf{\textsf{\textsf{DynFO(LFP)}}}$).
Intuitively, it permits storage of a polynomial amount of information
(generated in polynomial time), so that (for some problem) the information
during any update can be modified in $\textsf{P}$. Observe that the
above problem of $\textsc{lis(A)}$ lies in $\textsf{DynP}$, since
at every stage we stored a polynomial amount of information, and the
\emph{update }step took polynomial time to modify the information.

Although we do not consider relations between static and dynamic complexity
classes here, it is worth mentioning that $\textsf{DynP=P}$ (under
a suitable notion of a reduction). Hence, unless $\textsf{P=NP}$,
it is not possible to store some polynomial amount of information
(generated in polynomial time), so that insertion of a single edge
in a graph or a single clause in a 3-SAT expression (over a fixed
set of variables), leads to finding whether the graph is 3-colorable
or whether the 3-SAT expression is satisfiable. As another illustration,
for the $\textsf{NP}$-complete problem of finding the longest path
between any two vertices in an $n$-vertex undirected graph, even
if we are given any kind of (polynomial) information%
\footnote{By polynomial information, we mean information that has been generated
in polynomial time, and after the insertion of an edge, it can be
regenerated (in polynomial time) so as to allow insertion of another
edge, and so on ad infinitum.%
}, including the longest path between all possible pairs of vertices
in the given graph, it is not possible to find the \emph{new} longest
path between any pair of vertices when a single edge is inserted to
the graph (unless $\textsf{P=NP}$). This means that $\textsf{NP}$-complete
problems are even hard to simply update, i.e, even a small update
to an $\textsf{NP}$-complete problem cannot be done in polynomial
time. The reader is referred to \cite{hesse2003dynamic} for complete
problems for $\textsf{DynFO}$ and for reductions among problems in
the dynamic setting.

Although a dynamic programming solution to any problem is in effect
a $\textsf{DynP}$ solution, the class $\textsf{DynP}$ is less interesting
since it is essentially same as $\textsf{P}$. More interesting classes
are primarily the dynamic versions of smaller circuit complexity classes
inside $\textsf{P}$, like $\textsf{DynNC}^{1}$, $\textsf{DynTC}^{0}$,
etc. The most interesting, and perhaps the smallest dynamic complexity
class, is $\textsf{\textsf{\textsf{DynFO}}}$. Intuitively, $\textsf{\textsf{\textsf{DynFO}}}$
or Dynamic First-Order is the set of problems for which a polynomial
sized database of information can be stored to answer the problem
query (like reachability), such that after any insertion/deletion
of a tuple, the database can be updated using merely a $\textsf{FO}$
query (i.e. in First-Order Logic). A problem being in $\textsf{\textsf{\textsf{DynFO}}}$
means that any updation to the problem is extremely easy in some sense. 

Another reason why $\textsf{\textsf{\textsf{DynFO}}}$ is important
is because it is closely related to practice. A limitation of static
complexity classes is that they are not appropriate for systems where
large amounts of data are to be queried. Most real-life problems are
dynamic, extending over extremely long periods of time, manipulating
stored data. In such systems, it is necessary that small perturbations
to massive quantities of data can be computed very fast, instead of
processing the data from scratch. Consider for instance, a massive
code that is dynamically compiled. We would expect that the compilation,
as letters are typed, should be done very fast, since only a small
part of the program is modified with every letter. Hence, for huge
continually changing databases (or Big-Data), it is not feasible to
re-compute a query all over again when a new tuple is inserted or
deleted to/from the database. For the problems in $\textsf{\textsf{\textsf{DynFO}}}$,
since an SQL query is essentially a $\textsf{FO}$ Query, an SQL query
can update the database without computing everything again. This is
very useful in dynamic settings. A nice exposition on $\textsf{\textsf{\textsf{DynFO}}}$
in this respect can be found in \cite{schwentick2013perspectives}. 

One basic problem considered in this setting is that of Reachability.
In \cite{patnaik1994dyn}, it was shown that Undirected Reachability
(which is in the static class $\textsf{L}$), lies in the complexity
class $\textsf{\textsf{\textsf{DynFO}}}$. Note how a simple class
like $\textsf{FOL}$, which does not even contain parity, becomes
powerfully expressive in the dynamic setting. Hesse \cite{hesse2002dynamic}
showed that Directed Reachablity lies in $\textsf{DynTC}^{0}$. Also,
Dong and Su \cite{dong1995incremental} further showed that Directed
Rechability for acyclic graphs lies in $\textsf{\textsf{\textsf{DynFO}}}$. 

The Graph Isomorphism problem (of finding a bijection between the
vertex sets of two graphs such that the adjacencies are preserved)
has so far been elusive to algorithmic efforts and has not yet yielded
a better than subexponential ($2^{o(n)}$) time static algorithm.
The general problem is in $\textsf{NP}$, and also in $\textsf{SPP}$
(Arvind and Kurur \cite{arvind2002graph}). Thus, various special
cases have been considered, one important case being restriction to
planar graphs. Hopcroft and Wong \cite{hopcroft1974linear} showed
that Planar Graph Isomorphism can be decided in linear time. In a
series of works, it was further shown that Tree Isomorphism is in
$\textsf{L}$ (Lindell \cite{lindell1992logspace}), 3-connected Planar
Graph Isomorphism is in $\textsf{L}$ (Datta et. al. \cite{datta20083})
and finally, Planar Graph Isomorphism is in $\textsf{L}$ (Datta et.
al. \cite{datta2009planar}). 

Etessami considered the problem of isomorphism in the dynamic setting.
It was shown in \cite{etessami1998dynamic} that Tree Isomorphism
can be decided in $\textsf{\textsf{\textsf{DynFO}}}$. 

In this work, we consider a natural extension and show that isomorphism
for Planar 3-connected graphs can be decided in $\textsf{\textsf{\textsf{DynFO}}}$
(with some polynomial precomputation). Our method of showing this
is different from that in \cite{etessami1998dynamic}. The main technical
tool we employ is that of Canonical Breadth-First Search trees (abbreviated
CBFS tree), which were used by Thierauf and Wagner \cite{thierauf2010isomorphism}
to show that 3-connected Planar Graph Isomorphism lies in $\textsf{UL}$.
We also introduce a novel method for finding the canon of a 3-connected
Planar graph from Canonical Breadth-First Search trees in First-Order
Logic ($\textsf{FOL}$). We finally compare the canons of the two
graphs to decide on isomorphism.\\

Our main results are:
\begin{enumerate}
\item Breadth-First Search for undirected graphs is in $\textsf{\textsf{\textsf{DynFO}}}$
\item Isomorphism for Planar 3-connected graphs is in $\textsf{\textsf{\textsf{DynFO}}}^{+}$
\end{enumerate}
$\textsf{\textsf{\textsf{DynFO}}}^{+}$ is exactly same as $\textsf{\textsf{\textsf{DynFO}}}$,
except that it allows some polynomial precomputation, which is necessary
until enough edges are inserted so that the graph becomes 3-connected.
Note that this is the best one can hope for, due to the requirement
of 3-connectivity. 

In Section 2, we give the preliminary definitions and necessary explanations.
In sections 3 and 4, we prove Result 1. In Section 5, we prove Result
2. In Section 6, we introduce a novel method of canonizing a planar
3-connected graph in $\textsf{FOL}$ from Canonical Breadth-First
Search trees. In section 7, we conclude with open problems and scope
for future work.

\section{Preliminaries}

\subsubsection*{I. On Graph Theory: }

The reader is referred to \cite{diestel2005graph} for the graph-theoretic
definitions in this section.

A graph $G=(V,E)$ is connected if there is a path between any two
vertices in $G$. A pair of vertices $u,v\in V$ is a separating pair
if $G(V\backslash\{u,v\})$ is not connected. A graph with no separating
pairs is 3-connected.

\label{preliminaries-I}Let $E_{v}$ be the set of edges incident
to $v$. A permutation $\pi_{v}$ on $E_{v}$ that has only one cycle
is called a rotation. A rotation scheme for a graph G is a set $\pi$
of rotations, $\pi=\{\pi_{v}\ |\ v\in V$ and $\pi_{v}$ is a rotation
on $E_{v}\}$. Let $\pi^{c}$ be the set of inverse rotations, $\pi^{c}=\{\pi_{v}^{c}\ |\ v\in V\}$.
A rotation scheme $\pi$ describes an embedding of graph $G$ in the
plane. $\pi$ is a \emph{planar rotation scheme} if the embedding
is planar.

A planar graph $G$, along with its planar embedding (given by $\pi$)
is called a plane graph $G=(G,\pi)$. A plane graph divides the plane
into regions. Each such region is called a face.

For 3-connected planar graphs, we shall asssume that $\pi$ is the
set of anti-clockwise rotations around each vertex, and $\pi^{c}$
is the set of clockwise rotations around every vertex. Whitney \cite{whitney1933set}
showed that $\pi$ and $\pi^{c}$ are the only two rotations for 3-connected
planar graphs.

Two graphs $G=(V_{G},E_{G})$ and $H=(V_{H},E_{H})$ are said to be
isomorphic $(G\cong H)$ if there is a bijection $\phi:V_{G}\rightarrow V_{H}$
such that $(u,v)\in E_{G}\Leftrightarrow(\phi(u),\phi(v))\in E_{H}$.

\subsubsection*{II. On Finite-Model Theory:}

Please refer to any text, like \cite{immerman1999descriptive} for
the definitions on Finite-Model Theory.

A \emph{vocabulary} $\tau$ = $\langle R_{1}^{a_{1}},...,R_{r}^{a_{r}},c_{1},...,c_{s}\rangle$
is a tuple of relation symbols and constant symbols. 

A \emph{structure} $A$ over $\tau$ is a tuple, $A=\langle|A|,R_{1}^{a_{1},A},...,R_{r}^{a_{r},A},c_{1}^{A},...,c_{s}^{A}\rangle$,
where $|A|=\{0,1,2,...,n-1\}$ is a fixed size universe of size $||A||=n$. 

Let $STRUC(\tau)$ denote all possible structures over $\tau$, then
$S\subseteq STRUC(\tau)$ is any complexity theoretic problem. Let
$S\subseteq STRUC(\sigma)$ and $T\subseteq STRUC(\tau)$ be two problems,
where $\tau$ and $\sigma$ are two vocabularies.

A First-Order (FO) query $I:STRUC(\sigma)\rightarrow STRUC(\tau)$
is a tuple of $r+s+1$ formulas, $\langle\varphi_{0}\ldots\varphi_{r},\psi_{1}\ldots\psi_{s}\rangle$. 

For each $A\in STRUC(\sigma)$, 
\[
I(A)=\langle|I(A)|,\ R_{1}^{a_{1},I(A)}...,R_{r}^{a_{r},I(A)},\ c_{1}^{I}(A),...,c_{s}^{I}(A)\rangle
\]
where $|I(A)|=\{\langle b^{1},...,b^{k}\rangle\mid A\models\varphi_{0}(b^{1},...,b_{k})\}$, 

\[
R_{i}^{I(A)}=\{(\langle b_{1}^{1},...,b_{1}^{k}\rangle,...,\langle b_{a_{i}}^{1},...,b_{a_{i}}^{k}\rangle)\in|I(A)|^{a_{i}}\text{ such that }A\models\varphi_{i}(b_{a_{i}}^{1},...,b_{a_{i}}^{k})\}
\]
 
\[
c_{i}^{I(A)}\text{ = the unique }\langle b^{1},...,b^{k}\rangle\in|I(A)|\text{ such that }A\models\psi_{i}(b^{1},...,b^{k})
\]
 letting the free variables of $\varphi_{i}$ be $x_{1}^{1},...,x_{1}^{k},...,x_{a_{i}}^{1},...,x_{a_{i}}^{k}$,
and of $\varphi_{0}$ and $\psi_{j}$'s be $x_{1}^{1},...,x_{1}^{k}$. 

We shall refer to the following theorem at certain places, and we
make it explicit here, which can be proven using Ehrenfeucht-Fraisse
Games \cite{immerman1999descriptive}:
\begin{thm}
\label{thm-Transitive-Closure-not-in-fol}Transitive Closure is not
in $\textsf{FOL}$ = uniform $\textsf{AC}^{0}$. 
\end{thm}

\subsubsection*{III. On Dynamic Complexity: }

Refer to \cite{immerman1999descriptive} or \cite{patnaik1994dyn}
for the following definition and relevant examples:
\begin{defn}
\label{def-dynamiccomplexityclass}For any static complexity class
$\textsf{C}$, we define its dynamic version, $\textsf{DynC}$ as
follows: Let $\rho=\langle R_{1}^{a_{1}},...,R_{s}^{a_{s}},c_{1},...,c_{t}\rangle$,
be any vocabulary and $S\subseteq STRUC(\rho)$ be any problem. Let
${\displaystyle R_{n,\rho}=\{ins(i,a'),\ del(i,a'),\ set(j,a)\ |\ 1\le i\le s,\ a'\in\{0,...,n-1\}^{a_{i}},\ 1\le j\le t\}}$
be the request to insert/delete tuple $a'$ into/from the relation
$R_{i}$, or set constant $c_{j}$ to $a$.
\end{defn}
Let $eval_{n,\rho}:R_{n,\rho}^{*}\ \rightarrow STRUC(\rho)$ be the
evaluation of a sequence or stream of requests. Define $S\in\textsf{DynC}$
iff there exists another problem $T\subset STRUC(\tau)$ (over some
vocabulary $\tau$) such that $T\in\textsf{C}$ and there exist maps
$f$ and $g$:

\[
f:R_{n,\rho}^{*}\rightarrow STRUC(\tau),\ g:STRUC(\tau)\times R_{n,\rho}\rightarrow STRUC(\tau)
\]

satisfying the following properties:
\begin{enumerate}
\item \textbf{\emph{(Correctness)}} For all $r'\in R_{n,\rho}^{*}$, $(eval_{n,\rho}(r')\ \in\ S)\Leftrightarrow(f(r')\in T)$
\item \textbf{\emph{(Update)}} For all $s\ \in R_{n,\rho}$, and $r'\in R_{n,\rho}^{*}$,
$f(r's)=g(f(r'),s)$
\item \textbf{\emph{(Bounded Universe) }}$||f(r')||\ =\ ||eval_{n,\rho}(r')||^{O(1)}$
\item \textbf{\emph{(Initialization)}} The functions $g$ and the initial
structure $f(\emptyset)$ are computable in $\textsf{C}$ as functions
of $n$.
\end{enumerate}
Our main aim is to define the update function $g$ (over some vocabulary
$\tau$). If condition (4) is relaxed, to the extent that the initializing
function $f$ may be polynomially computable (before any insertion
or deletion of tuples begin), the resulting class is $\textsf{DynC}^{+}$,
that is $\textsf{DynC}$ with polynomial precomputation.

\subsubsection*{IV. \label{preliminaries-part-IV}On $\textsf{DynFO}^{+}$:}

Here we shall explain the polynomial precomputation part of Definition
\ref{def-dynamiccomplexityclass} above with respect to the problem
of 3-connected Planar graph isomorphism.

Condition 4 in the Definition \ref{def-dynamiccomplexityclass} above
requires that the function $f$ be computable in the static complexity
class $\textsf{C}$ as a function of $n$. Relaxing that condition
implies that for static complexity classes $\textsf{C}$ that are
contained in $n^{O(1)}$, $f$ may not be computable in $\textsf{C}$,
but is atleast computable efficiently, i.e. polynomially or in $\textsf{FO(LFP)}$.

In the dynamic setting, edges are added (or removed) at every stage.
As such, the graph at any stage is either both planar and 3-connected
(state $\emph{A}$), or is neither planar nor 3-connected nor both
(state $\emph{B}$). Since we assume the conditions of planarity and
3-connectivity, our relations do not hold in state $\emph{B}$, but
only in state $\emph{A}$. 

We shall maintain a tuple of relations for the problem, call it $T$,
which need to be populated at every stage, be it $\emph{A}$ or $\emph{B}$.
When the edge insertion process begins for the first time, the graph
will be empty and in state $\emph{B}$. As edges are added and removed
from the graph, it will stay in state $\emph{B}$ until there are
sufficient edges to satisfy the constraints of 3-connectivity and
planarity, which will put the graph in state $\emph{A}$. Uptil now,
the computation for relations in $\emph{T}$ could not be done in
$\textsf{FOL}$, and as such, the relations cannot be maintained in
$\textsf{\textsf{\textsf{\textsf{\textsf{DynFO}}}}}$. Hence, during
state $\emph{B}$, the relations in $\emph{T}$ must be maintained
using polynomial queries, or queries in $\textsf{FO(LFP)}$. This
can easily be done since the problem of Planar Graph Isomorphism itself
is in Logspace \cite{datta2009planar}, and we omit its details. Hence,
polynomial precomputation is necessary for the function $f$ as in
Definition \ref{def-dynamiccomplexityclass} above when the graph
is in state $\emph{B}$, until it reaches state $\emph{A}$. Also,
if ever the graph goes into state $\emph{B}$ during insertions and
deletions, we again need to resort to polynomial queries.

Also, since no known algorithm exists in $\textsf{\textsf{\textsf{\textsf{\textsf{DynFO}}}}}$
to \emph{decide} whether a given graph is 3-connected and planar,
even this needs to be done polynomially.

Once the graph is in state $\emph{A}$ or both planar and 3-connected,
we will show the existence of $\emph{T}$ such that the canonical
description of the graph can be maintained in $\textsf{\textsf{\textsf{\textsf{\textsf{DynFO}}}}}$,
i.e. the canonical description can be maintained in $\textsf{FOL}$
for insertions/deletions of edges as long as the graph is in state
$\emph{A}$.

\subsubsection*{V. On Conventions:}

Throughout this paper, we adopt the following convention: if $R$
is any relation or any of our denotation, $R'$ denotes the updated
relation, or the denotation in the updated relation. Also, for any
query $Q$, $Eq(Q)$ will denote the equivalent query formed by replacing
all the $a$'s in $Q$ by $b$'s, and all the $b$'s in $Q$ by $a$'s.
For the ease of readability, we shall only write the queries in a
high-level form, and leave out their easy translation to the exact
form (which quickly turns non-elegant and lengthy). 

We may often use a statement of the form $\alpha\leftarrow\beta$,
i.e. we are assigning the value of $\beta$ to $\alpha$. In some
cases, we can only deal with relations, or sets of tuples and not
individual variables. Hence, we do this by creating some temporary
relation (say) $temp$, which contains only 1 element, i.e. $\alpha$.
Note that $\beta$ may itself be a first-order formula, or a first-order
statement whose result is just one element in the universe, i.e. $\alpha$.
After that, wherever we need to use $\alpha$, we use $\forall x,\ temp(x)$,
and since $temp$ contains only one element, i.e. $\alpha$.

For brevity of expressing relations, we may often use the following
short-hand notation, here shown for a relation $R$ of arity 4: 
\[
R(a_{1},a_{2},\{b_{1},b_{2},...,b_{k}\},a_{4})=\bigwedge_{i=1}^{k}R(a_{1},a_{2},b_{i},a_{4})
\]

Moreover, to prevent notational clutter from interfering with the
general conceptual flow of the paper, we relegate all queries to the
Appendix.

\section{Ordering and Arithmetic in $\textsf{\textsf{\textsf{DynFO}}}$}

In this section, we prove that both Ordering and Arithmetic can be
done in $\textsf{\textsf{\textsf{\textsf{DynFO}}}}$. As such, an
explicit order on the universe in the structures of the vocabulary
is not needed when working in $\textsf{\textsf{\textsf{\textsf{DynFO}}}}$. 

We will need to maintain the running sums of shortest paths in the
graph for which we will need the basic arithmetic operations of addition
and subtraction. Since transitive closure is not in $\textsf{FOL}$
(see Theorem \ref{thm-Transitive-Closure-not-in-fol}), we cannot
add a set of arbitrary numbers in $\textsf{FOL}$. But we can add
a set of $\emph{k}$ numbers in $\textsf{FOL}$, where $\emph{k}$
is a constant.

The crucial thing to note is that arithmetic is only as good as the
ordering. By this, we mean that any query for addition or subtraction
can be converted to an equivalent query for ordering. For example,
querying $4+7$, for addition, is equivalent to querying: the $7$th
element in the ordering from the ($4$th element in the ordering from
the (\emph{least ordered element})). 

Also, the build-up of ordering and summations needs to be done only
during \emph{insertion} queries. Moreover, the relations developed
in this section hold for \emph{both} the states $\emph{A}$ and $\emph{B}$
as described in the part IV of Preliminaries (\ref{preliminaries-part-IV}).

The ideas developed here are similar to the ones in \cite{etessami1998dynamic},
specifically that the Ordering can be maintained in $\textsf{\textsf{\textsf{\textsf{DynFO}}}}$.
We essentially show the same thing, except extending the fact that
these relations hold even for arbitrarily large graphs.

\subsection{Maintaining the Universe in $\textsf{\textsf{\textsf{\textsf{DynFO}}}}$}

One subtle fact that needs to be considered are the universes on which
the operators $\forall$ and $\exists$ are going to act. Either we
can explicitly maintain the universes on which the operators will
act, or we can choose the universes to be the same as the ones for
the input. Here we shall explicitly maintain the universe, using a
unary relation $U(x)$ which holds if $x$ belongs to the universe.
Only previously unknown elements will be allowed to enter the universe.
Note that since we are maintaining our own universe, the universe
from which the symbols are picked need not have finite size. The universe
can grow arbitrarily large, and we will maintain the necessary ordering
and summations for it (shown in further sections). 

The queries in a high-level form to maintain the Universe, or $U(x)$,
during $insert(a,b)$ are as follows (we give these relations to give
an idea of the manner in which queries will be expressed throughout):
\\

$U'(x)=U(x)\vee(\neg U(a)\wedge x=a)$

$U''(x)=U(x)\vee(\neg U(b)\wedge x=b)$

\subsection{Ordering in $\textsf{\textsf{\textsf{\textsf{DynFO}}}}$}

To maintain Ordering in $\textsf{\textsf{\textsf{\textsf{DynFO}}}}$,
we shall maintain the relation $O(x,y)$ which will be the transitive
closure on the ordering relation, implying $x\leq y$. 

The (total) order will be decided on the basis of the \emph{first}
time a specific element in the universe is used as some part of an
``insertion'' query. This means that the first time some tuple (edge
in our case) is added to the graph which contains the specific element,
that element will enter the ordering relation. 

For instance, when the graph is empty, if the first query is to insert
an edge between the vertices numbered $(9,7)$, i.e. $insert(9,7)$,
we add 9 < 7 to the ordering, meaning we insert $(9,7)$ in $O$.
If the second query is $insert(4,7)$, first we add $4$ to the ordering
relation; hence, we add $(9,4)$ and $(7,4)$ to $O$. Since $7$
(the second element of the query) is already present in the ordering
relation, we do nothing. If the next query is $insert(4,9)$, we still
do nothing since both $4$ and $9$ are in the ordering relation.
If the query after that is $insert(4,8)$, since $4$ is already in
the ordering relation, we do nothing; but 8 is not in the relation.
Hence, we add the following tuples to O: $(9,8)$, $(7,8)$, and $(4,8)$.
Also, for each new element, say $8$ in this case, we add $(8,8)$
to the ordering relation to satisfy the equality too. At this juncture,
the ordering that we have is: $9<7<4<8$. Note that every number here
is treated as a \emph{symbol} and the symantic value of the number
is ignored. 

The queries in a high-level form to maintain the ordering, or $O(x,y)$,
during $insert(a,b)$ are as follows: \\

$O'(x,y)=O(x,y)\vee(\neg U(a)\wedge y=a)$ 

$O''(x,y)=O'(x,y)\vee(\neg U'(b)\wedge y=b)$

\subsection{Summation in $\textsf{\textsf{\textsf{DynFO}}}$}

We shall use the ordering relation to maintain the summations in $\textsf{\textsf{\textsf{\textsf{\textsf{DynFO}}}}}$.
But to start with the summations, we need a minimum element (or the
identity for summation), the $0$th unique element, say $min$, for
which the following will hold: $min+min=min$. We shall choose $min$
as the \emph{first} element that enters the ordering. The relation
that we will maintain will be $Sum(t,x,y)$, which would hold for
any $t,x$ and $y$ if $t=x+y$. We shall now show that this relation
can be maintained in $\textsf{\textsf{\textsf{\textsf{\textsf{DynFO}}}}}$. 

Assume that the ordering due to some sequence of insertions and deletions
is as follows: $9<7<4<8<3<5<1<2<6$. Our summations, will satisfy
the following invariant: \emph{the $i$'th element in the ordering
+ the $j$'th element in the ordering = $(i+j)$'th element in the
ordering}. Thus 9 will be the 0'th element in the above example. Hence
$7+4=8$, $9+7=7$, $8+5=6$ etc. Also, note that $1+2$ would be
the $(6+7)$'th element, or the $13$'th element; But the ordering
does not still have the $13$'th element, and as such, this summation
does not exist. Hence, by maintaining the summation in $\textsf{\textsf{\textsf{\textsf{\textsf{DynFO}}}}}$,
we will need that our summations remain below the largest ordered
element in the ordering relation. We will see that this will suffice
for the problem that we have at hand. Moreover, we assume from here
onwards that whenever we use some number, like 0 or 1 or any other
explicitly in this paper, we mean the 0'th or the 1'th or corresponding
element in the ordering. 

Once we have the summation relation, the differences between numbers
can also be easily maintained. By this, we mean that a query for $t=x-y$
is equivalent to $Sum(x,t,y)$. Hence, we shall not explicitly maintain
the difference relation.

The main idea in updating summation is that if we have the list of
all the 2-tuples which sum to the maximum element, then this list
can be used to create, in $\textsf{FOL}$, the list of all the 2-tuples
which sum to one more than the maximum element. 

We insert the elements $a$ and $b$ if they are not present in the
universe, and sequentially make the change. 

\label{to-summation-dynfo}The queries in a high-level form to maintain
the summation, or $Sum(t,x,y)$, during $insert(a,b)$, when for the
case of $a$ being inserted in the universe are given in \ref{alg-summations-dynfo}
\\

From now onwards, we shall freely use the notations $a\leq b$ and
$a+b$ whenever necessary, knowing that they can be easily replaced
by $O(a,b)$ and $\forall t,\, Sum(t,a,b)$.

We can conclude from this section that:
\begin{thm}
Ordering and Arithmetic can be maintained in $\textsf{\textsf{\textsf{\textsf{\textsf{DynFO}}}}}$. 
\end{thm}

\section{Breadth-First-Search in $\textsf{\textsf{\textsf{\textsf{DynFO}}}}$}

In this section, we shall show that Breadth-First-Search (abbreviated
\textsc{BFS}) for any arbitrary undirected graph lies in $\textsf{\textsf{\textsf{\textsf{DynFO}}}}$.
More specifically, we shall show that there exists a set of relations,
such that using those relations, finding the minimum distance between
any two points in a graph can be done through \textsf{FOL}, and the
set of all the points at a particular distance from a given point
can be retrieved through a \textsf{FO} query, in any arbitrary undirected
graph. Also, the modification of the relations can be carried out
using \textsf{FOL}, during insertion or deletion of edges.

The definitions and terminologies regarding \textsc{BFS} can be found
in any standard textbook on algorithms, like \cite{Cormen:2001:IA:580470}.

The main idea is to maintain the \textsc{BFS} tree from \emph{each}
vertex in the graph. This idea is important, because it will be extended
in the next section. To achieve this, we shall maintain the following
relations:
\begin{itemize}
\item $Level(v,x,l)$, implying that the vertex $x$ is at level $l$ in
the \textsc{BFS} tree of vertex $v$ (A vertex $x$ is said to be
at level $l$ in the \textsc{BFS} tree of $v$ if the distance between
$x$ and $v$ is $l$);
\item $BFSEdge(v,x,y)$, meaning that the edge $(x,y)$ of the graph is
in the \textsc{BFS} tree rooted at $v$;
\item $Path(v,x,y,z)$ meaning that vertex $z$ is on the path from $x$
to $y$, in \textsc{BFS} tree of $v$. Also
\item $Edge(x,y)$ will denote all the edges present in the entire graph.
\end{itemize}
Note that it is sufficient to maintain the $Level$ relation to query
the length of the shortest path between any two vertices. We maintain
the $BFSEdge$ and $Path$ relations only if we want the actual shortest
path between any two vertices.

These relations form the vocabulary $\tau$ as in Definition \ref{def-dynamiccomplexityclass}.

\subsection{Maintaining $Edge(x,y)$}

Maintaning the edges in the graph is trivial.

During insertion: $Edge'(x,y)=Edge(x,y)\vee(x=a\wedge y=b)\vee(x=b\wedge y=a)$

During deletion: $Edge'(x,y)=Edge(x,y)\wedge(x\not=a\wedge y\not=b)\wedge(x\not=b\wedge y\not=a)$

\subsection{Maintaining $Level(v,x,l)$, $BFSEdge(v,x,y)$, $Path(v,x,y,z)$}

We shall first focus on the $Level(v,x,l)$ relation, since it will
give us the tools required for the other two relations.

In maintaining this relation, we are effectively maintaining the shortest
distances between every pair of vertices. We will need to understand
how the various \textsc{BFS} trees behave during insertion and deletion
of edges before we write down the queries.

We will use the following notations from this section onwards. Let
$path_{v}(\alpha,\beta)$ denote the set of edges in the path from
vertex $\alpha$ to $\beta$, in the \textsc{BFS} tree of $v$. Let
$|path_{v}(\alpha,\beta)|$ denote its size. Hence $Level(v,x,l)$
means $|path_{v}(v,x)|=l$. Let $level_{v}(x)$ denote the level of
vertex $x$ in \textsc{BFS} tree of $v$. Hence, $Level(v,x,l)\Leftrightarrow level_{v}(x)=l$.
Also, we shall succinctly denote the edge from $a$ to $b$ by $\{a,b\}$.
The vertices which are not connected to $v$ will not appear in any
tuple in the BFS-tree of $v$.

Note that any path can be split into two disjoint paths. For instance,
$path_{v}(a,b)=path_{v}(a,d)\cup path_{v}(d,b)$ for any vertex $d$
on $path_{v}(a,b)$, simply because there is only one path in a tree
between any two vertices.

\subsubsection{$insert(a,b)$}

Due to the insertion of edge $\{a,b\}$, various paths in many BFS
trees will change. We will show that many of the paths do not change,
and these can be used to update the shortest paths that do change. 

We shall see how to modify level of some vertex $x$ in the \textsc{BFS}
tree of some vertex $v$. But before we proceed, we'll need the following
important lemma:
\begin{lem}
\label{lem-level-change-during-insertion}After the insertion of an
edge $\{a,b\}$, the level of a vertex $x$ cannot change both in
the \textsc{BFS} trees of $a$ and $b$.\end{lem}
\begin{proof}
\begin{figure}[H]
\begin{centering}
\includegraphics[scale=0.5]{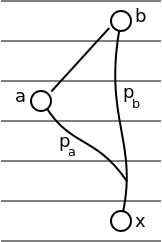}
\par\end{centering}

\caption{\label{fig1:Path-invariance-during-insertion}Path invariance during
insertion of an edge $\{a,b\}$}

\end{figure}
(Refer to Figure \ref{fig1:Path-invariance-during-insertion}) Before
the insertion of $\{a,b\}$, let $p_{a}=path_{a}(a,x)$, and $p_{b}=path_{b}(b,x)$.
Without loss of generality, let $|p_{a}|\leq|p_{b}|$. Now we can
show that $b$ does not lie on $p_{a}$. If $b$ does lie on $p_{a}$,
then $p_{a}=path_{a}(a,x)=path_{a}(a,b)\cup path_{a}(b,x)$, and $|p_{a}|=|path_{a}(a,b)|+|path_{a}(b,x)|$.
Since $a$ and $b$ are distinct vertices, $|path_{a}(a,b)|>0$, and
since $|path_{b}(b,x)|$ is the shortest path between the vertices
$b$ and $x$, $|path_{a}(b,x)|\geq|path_{b}(b,x)|$. Hence, $|p_{a}|=|path_{a}(a,b)|+|path_{a}(b,x)|>|path_{a}(b,x)|\geq|path_{b}(b,x)|=|p_{b}|$
which is a contradiction. Hence, $b$ cannot lie on $p_{a}$.

Since $|p_{a}|\leq|p_{b}|$ and $b$ does not lie on $p_{a}$, the
insertion of edge $\{a,b\}$ does not change the shortest distance
between $a$ and $x$. If the shortest distance changes, then the
new path will be $\{a,b\}\cup path_{b}(b,x)$, and $1+|p_{b}|<|p_{a}|$,
which would be a contradiction. A similar thing happens if $|p_{b}|\leq|p_{a}|$.
Hence, for every $x$, either $p_{a}$ or $p_{b}$ remains unchanged.
\end{proof}
Since the level of vertex $x$ remains invariant in atleast one \textsc{BFS}
tree, this fact can be used to modify the level of (and subsequently
even the paths to) $x$ using this invariant. This fact will be crucial
in the queries that we write next.

To update the $BFSEdge$ and $Path$ relations, since we will create
the new shortest path by joining together two different paths, we
need to ensure that these paths are disjoint. 

Without loss of generality, let $|path_{b}(b,x)|\leq|path_{a}(a,x)|$
\begin{lem}
\label{lem2}If any vertex $t$ is on $path_{b}(b,x)$ and on $path_{v}(v,a)$,
then the shortest path from $v$ to $x$ does not change after insertion
of the edge $\{a,b\}$\end{lem}
\begin{proof}
Consider Figure \ref{fig2:Disjointness-of-2paths-during-insertion}.
\begin{figure}[H]
\begin{centering}
\includegraphics[scale=0.5]{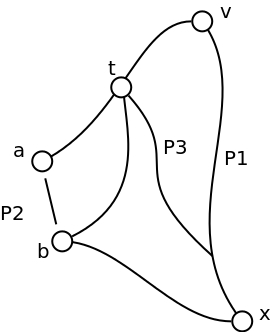}
\par\end{centering}

\caption{\label{fig2:Disjointness-of-2paths-during-insertion}Disjointness
of $path_{v}(v,a)$ and $path_{b}(b,x)$}

\end{figure}

Let $P_{1}=path_{v}(v,x)$. Let $P_{2}$ be the walk formed by the
concatenation of $path_{v}(v,a)$, the edge $\{a,b\}$, and $path_{b}(b,x)$
(the walk is a collection of edges, possibly repeated). Let $P_{3}$
be the walk formed by the concatenation of $path_{v}(v,t)$ and $path_{b}(t,x)$.
Since the walk $P_{3}$ was present even before the insertion of edge
$\{a,b\}$, $|P_{3}|\geq|P_{1}|$. Since $P_{3}$ is a subset of $P_{2}$,
$|P_{2}|\geq|P_{3}|$. Hence, $|P_{2}|\geq|P_{1}|$, as required.
\end{proof}
\label{emp-bfs-insert}The queries during insertion of edge $\{a,b\}$
are as given and illustrated in \ref{alg-bfs-insert}, and illustrated
in \ref{figA:-bfs-insert}. The correctness of the queries follows
from Lemma \ref{lem-level-change-during-insertion} and Lemma \ref{lem2}.

\subsubsection{$delete(a,b)$}

Consider now the deletion of some edge $\{a,b\}$ from the graph.
If it is present in the BFS tree of some vertex $v$, the removal
of the edge splits the tree into two different trees. Let $R_{1}=\{u\ |\ v,u\mbox{ are connected in }V_{G}\backslash\{a,b\}\}$,
and $R_{2}=\{u\ |\ u\notin R_{1}\}$. We find the set $PR=\{(p,r)\ |\ p\in R_{1}\wedge r\in R_{2}\wedge Edge(p,r)\}$,
where $PR$ is the set of edges in the graph that connect the trees
$R_{1}$ and $R_{2}$. The new path to $x$ will be a path from $v$
to $p$ in the BFS-tree of $v$, edge $\{p,r\}$, and path from $r$
to $x$ in the BFS-tree of $r$; and $\{p,r\}$ will be chosen to
yield the shortest such path, and we will choose $\{p,r\}$ to be
the lexicographically smallest amongst all such edges that yield the
shortest path. 

The only thing we need to address is the fact that the path from $r$
to $x$ in the BFS tree of $r$ does not pass through the edge $\{a,b\}$. 
\begin{lem}
When an edge $\{a,b\}$ separates a set of vertices $R_{2}$ from
the \textsc{BFS} tree of $v$, and $r$ and $x$ are vertices belonging
to $R_{2}$, then $path_{r}(r,x)$ cannot pass through edge $\{a,b\}$\end{lem}
\begin{proof}
Refer to Figure \ref{fig3:Deletion-of-edge-bfs}. 
\begin{figure}[H]
\begin{centering}
\includegraphics[scale=0.5]{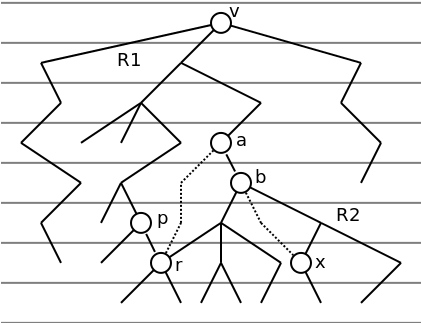}
\par\end{centering}

\caption{\label{fig3:Deletion-of-edge-bfs}Deletion of edge $\{a,b\}$ and
the subsequent changes (the lines represent levels in the BFS tree
of vertex $v$)}
\end{figure}
 We prove by contradiction. Let the shortest path from $r$ to $x$
pass through the edge $\{a,b\}$. For any two vertices $s$ and $t$,
$|path_{s}(s,t)|=|path_{t}(t,s)|$, though $path_{s}(s,t)$ may not
be equal to $path_{t}(t,s)$, since there can be only one \emph{value}
of shortest path in an \emph{undirected} graph, though any number
of such paths. Hence $|path_{r}(r,a)|=|path_{a}(a,r)|=1+|path_{b}(b,r)|=1+|path_{r}(r,b)|$.
Also, $|path_{a}(a,x)|=1+|path_{b}(b,x)|$. If $path_{r}(r,x)$ passes
through$\{a,b\}$, $|path_{r}(r,x)|=|path_{r}(r,a)|+|path_{a}(a,x)|=2+|path_{r}(r,b)|+|path_{b}(b,x)|$.
But $path_{r}(r,b)\cup path_{b}(b,x)$ is a path from $r$ to $x$,
and is shorter by atleast $2$ units from our claimed shortest path
from $r$ to $x$. Hence, $path_{r}(r,x)$ cannot pass through $\{a,b\}$. \end{proof}
\begin{rem}
An important observation is that \emph{the above lemma holds only
for the ``undirected'' case. It fails for the directed case, implying
that the same relations cannot be used for BFS in directed graphs}.
To see a simple counter-example, note that there can be a directed
edge from $r$ to $a$ in the directed case, and in that case, the
shortest path from $r$ to $x$ can pass through $(a,b)$.
\end{rem}
Also note that for every vertex $x$ in $R_{1}$, the shortest path
from $v$ to $x$ remains the same, since removal of an edge cannot
\emph{decrease} the shortest distance. 

\label{emp-bfs-delete}The queries during deletion of edge $\{a,b\}$
are given in \ref{alg-bfs-delete}. (Refer to Figure \ref{fig3:Deletion-of-edge-bfs}
for illustration on selection of edge $\{p,r\}$)
\begin{rem}
Note that although we pick the new paths for every vertex in the set
$R_{2}$ in parallel, we need to ensure that the paths picked are
consistent, i.e. the paths form a tree and no cycle is formed. This
is straightforward to see, since if a cycle is formed, it is possible
to pick another path for some vertex that came earlier in the lexicographic
ordering. Hence, our queries are consistent.
\end{rem}
This leads us to the following theorem:
\begin{thm}
Breadth-First-Search for an undirected graph is in $\textsf{\textsf{\textsf{\textsf{DynFO}}}}$.
\\

\end{thm}
\textbf{Note on the nature of deletion in }$\textsf{\textsf{\textsf{\textsf{DynFO}}}}$\textbf{:}

One thing that is not completely clear about the complexity class
$\textsf{\textsf{\textsf{\textsf{DynFO}}}}$ is that somehow maintaining
the relations in $\textsf{\textsf{\textsf{\textsf{DynFO}}}}$ during
insertion of edges (or tuples) is far easier than during deletion.
A case in point is the problem of reachability in directed graphs,
for which relations can be maintained that answer the problem query
during insertion of edges, but whether there exist relations that
can answer the query (reachability in directed graphs) during \emph{both}
insertion and deletion is an open problem. A particular reason for
this hardness arising during deletion maybe the fact that we \emph{always
begin with an empty graph}. Hence, we have all the information for
some point uptil an edge is inserted, which just adds some new information.
But when an edge is deleted, we would need all the information about
the graph that is formed from the empty graph by inserting edges upto
the point that this (deleted) edge is not inserted, but that may not
be possible using only polynomial space. To illustrate what we mean,
consider insertion of edges in the sequence - $e_{1},e_{2},e_{3},e_{4},e_{5}$.
Now if the edge $e_{3}$ is deleted, the relations do not have enough
information, i.e. they do not have the tuples that would have been
formed if the edges were instead inserted in the sequence $e_{1},e_{2},e_{4},e_{5}$
to an empty graph. We can try to go past this problem of \emph{monotonicity}
by storing all the information for all the possible sequences of insertion,
but that would require storing exponential amount of data, an we would
no longer remain in polynomial space, and thus not in $\textsf{\textsf{\textsf{\textsf{DynFO}}}}$.

Another thing that one may try is to begin with a complete graph instead
of an empty graph, and then try the deletion process. This would make
the deletion process extremely easy (it would act exactly as a dual
of insertion), but to begin with, we would need tuples required for
the \emph{entire }graph. This would itself require either a lot of
computation, and would beat the purpose of solving the problem in
$\textsf{\textsf{\textsf{\textsf{DynFO}}}}$, or exponential space,
in which case the problem would no longer be in $\textsf{\textsf{\textsf{\textsf{DynFO}}}}$.

\section{3-connected Planar Graph Isomorphism}

The ideas and the techniques hitherto developed implied for both the
states \emph{A} and \emph{B} as stated in part IV of Preliminaries
(\ref{preliminaries-part-IV}). Now onwards, our relations would no
longer hold for general graphs, and we restrict ourselves to 3-connected
and planar graphs, or state \emph{$A$} mentioned therein.

We shall now show how to maintain a canonical description of a 3-connected
planar graph in $\textsf{\textsf{\textsf{\textsf{DynFO}}}}$. To achieve
this end, we shall maintain Canonical Breadth-First Search (abbreviated
\textsc{CBFS}) trees similar to the ones used by Thierauf and Wagner
\cite{thierauf2010isomorphism}.

\subsection{Canonical Breadth-First Search Trees}

We shall define \textsc{CBFS} trees here for the sake of completeness.
A theorem of Whitney \cite{whitney1933set} says that 3-connected
planar graphs have a unique embedding on the sphere. Intuitively,
it states a topological fact that there is one and only one way to
order the edges around each vertex if a graph has to remain 3-connected
and planar. This fact is used by Thierauf and Wagner to construct
the \textsc{CBFS} trees.

Consider a 3-connected planar graph and a vertex $v$ in it. Let $N(v)=\{u:\ Edge(v,u)\}$,
that is, $N(v)$ is the set of neighbours of vertex $v$. Let $D(v)=\{0,1,2,...,d_{v}-1\}$
where $d_{v}=|N(v)|$, the degree of the vertex $v$. Consider a permutation
$\pi_{v}:N(v)\rightarrow D(v)$, such that for some $v_{1},v_{2}\in N(v)$,
if $(v,v_{2})$ is $i$'th edge encountered while moving anti-clockwise
from the edge $(v,v_{1})$, then $\pi_{v}(v_{2})=(\pi_{v}(v_{1})+i)$
mod $d_{v}$. Also, as defined in the part I of Preliminaries (\ref{preliminaries-I}),
$\pi=\{\pi_{v}\ |\ v\in V_{G}\}$, and $\pi^{c}=\{\pi_{v}^{c}\ |\ v\in V_{G}\}$.

We now define the Canonical Breath-First Search method. We need a
designated vertex and an edge to begin with. Let $v$ be the vertex,
and $v_{e}$ be another vertex such that $(v,v_{e})$ is the designated
edge. The Canonical \textsc{BFS} tree thus formed will be designated
as $[v,v_{e}]$. Informally, we sequentially add the edges starting
from $(v,v_{e})$, and then moving from $v_{e}$ towards other vertices
as per the rotation $\pi_{v}$. For every other vertex $u$, we consider
the edges according to $\pi_{u}$, beginning with the edge $(u,u_{p})$,
where $u_{p}$ is the parent edge of $u$ in the \textsc{CBFS} tree.

\label{emp-cbfs-method}Formally, the pseudocode for $[v,v_{e}]$,
according to some $\pi$ is as described in \ref{alg-cbfs-method}. 

Conider figure \ref{fig4:Canonical-Breadth-First-Search}. If we had
to perform BFS starting from vertex $u$ as per the numbering around
every vertex as shown, the sequence of vertices visited around $u$
would be $x,v,w$. We would then start with $x,$ and the sequence
would be $w,y,u$, but since $w$ and $u$ are already visited, we
just visit $y$. 

The way the numbering around every vertex is chosen given the starting
edge $(u,x)$ is as follows: We first number vertices around $u$,
and hence the vertices $x,v,w$ get the numbers $0,1,2$ (or $1,2,3$)
around $u$. Next, for every vertex in the queue, we give its parent
the number $0$, and give increasing numbers to the remaining vertices
in anti-clockwise order. Hence, the numbering around $x$ would be
$u,w,y$ with the numbers $0,1,2$. The resulting numbering is as
shown in Figure \ref{fig7:Normalized-CBFS-tree}.

\begin{figure}
\begin{centering}
\includegraphics[scale=0.5]{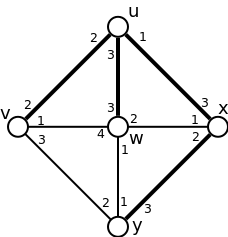}
\par\end{centering}

\caption{\label{fig4:Canonical-Breadth-First-Search}Canonical Breadth-First
Search Method for some anti-clockwise embedding, starting from $u$}

\end{figure}

We maintain a \textsc{CBFS} tree, denoted by $[v,v_{e}]$, from \emph{each}
vertex $v$ in the graph, for \emph{each} edge $(v,v_{e})$ used as
the starting embedding edge. This set of \textsc{CBFS} trees will
help us in maintaining the necessary relations, during insertions
and deletions, for isomorphism.

We need to slightly modify our previous conventions as used in the
case of \textsc{BFS} trees, since now the \textsc{CBFS} trees depend
not only on the vertices, but also on the chosen edge. Let $path_{v,v_{e}}(\alpha,\beta)$
denote the path from vertex $\alpha$ to $\beta$, in the \textsc{CBFS}
tree $[v,v_{e}]$. The other notations remain same.

Let Least Common Ancestor (LCA) of $x$ and $y$, $lca_{v,v_{e}}(x,y)$,
denote that vertex $d$ which is on $path_{v,v_{e}}(v,a)$ and $path_{v,v_{e}}(v,x)$
and whose level is \emph{maximum} amongst all such vertices. 

Denote the embedded number of vertex $x$ around vertex $u$ by $emnum_{u}(x)$,
i.e. $emnum_{u}(x)=\pi_{u}(x)$. Let $parent_{v,v_{e}}(u,u_{p})$
be true if $u_{p}$ is the parent of $u$ in $[v,v_{e}]$. Denote
by $emnum_{v,v_{e}}(u,x)=(\pi_{u}(x)-\pi_{u}(u_{p}))$ mod $d_{u}$,
the embedded number of some vertex $x$ around $u$ in $[v,v_{e}]$
if $u_{p}$ (the parent of $u$ in $[v,v_{e}]$) has been assigned
the number $0$ (or the minimum embedding number).

Since the embedding $\pi$ is unique, given a vertex and an edge incident
on it, the entire \textsc{CBFS} tree is fixed. As such, given $v,v_{e}$,
the \emph{length} of the shortest path to a vertex $x$ is fixed.
The actual path is decided by the following definition.
\begin{defn}
\label{def-can-order}Let $<_{c}$ denote a canonical ordering on
paths. Let $P_{1}=path_{v,v_{e}}(v,x_{1})$ and $P_{2}=path_{v,v_{e}}(v,x_{2})$.
Let $d=lca_{v,v_{e}}(x_{1},x_{2})$, $d_{1}$ a vertex on $P_{1}$
and $d_{2}$ a vertex on $P_{2}$, and $(d,d_{1})$ and $(d,d_{2})$
edges in the graph.

Then $P_{1}<_{c}P_{2}$ if: 
\begin{itemize}
\item $|P_{1}|<|P_{2}|$ or
\item $|P_{1}|=|P_{2}|$ and $emnum_{v,v_{e}}(d,d_{1})<emnum_{v,v_{e}}(d,d_{2})$\\

\end{itemize}
\end{defn}
We shall now see how to maintain the \textsc{CBFS} trees $[v,v_{e}]$.\\

We maintain the following relations:
\begin{itemize}
\item $Emb(v,x,n_{x})$, meaning that the vertex $x$ is in the neighbourhood
of $v$, and the edge $(v,x)$ around $v$ has the embedded number
$n_{x}$;
\item $Face(f,x,y,z)$, meaning that the vertex $z$ is in the anti-clockwise
path from vertex $x$ to vertex $y$, around the face labelled $f$.

Two things need to be said here. Since the number of faces in a planar
graph with $n$ vertices can be more than $n$, we should label the
face with a 2-tuple instead of a single symbol; but we do not do this
since it adds unnecessary technicality without adding any new insight.
If required, all the queries can be maintained for the faces labelled
as two tuples $f=(f_{1},f_{2})$.

Also, we maintain the transitive closure of edges around each face,
instead of simply the edges that constitute each face, since the splitting
and merging of faces cannot be maintained in $\textsf{\textsf{\textsf{\textsf{DynFO}}}}$
by just maintaining the set of edges. This will become clear later
in further sections that maintain the $Face$ relation;

\item $Level(v,x,l)$, meaning that the vertex $x$ is at level $l$ in
the \textsc{BFS} tree of $v$. This is exactly as in the general case.
\item $CBFSEdges(v,v_{e},s,t)$, where $(s,t)$ is an edge in the CBFS tree
$[v,v_{e}]$
\item $CPath(v,v_{e},x,y,z)$ denoting that $z$ is on the path from $x$
to $y$ in $[v,v_{e}]$
\end{itemize}

\subsection{Maintaining $Emb(v,x,n_{x})$ and $Face(f,x,y,z)$}

These two relations define the embedding of the graph in the plane.
Before proceeding, the following lemmas are required. Note that these
lemmas will also be required for maintaining other relations as mentioned
above. We further assume throughout in this section that the embedded
numbering is the anti-clockwise one, and note that the same relations
that we maintain for the anti-clockwise embedding can be maintained
for the clockwise embedding.
\begin{lem}
\label{lem-one-face-for-2-distict-vertices}In a 3-connected planar
graph $G$, two distinct vertices not connected by an edge cannot
both lie on two distinct faces unless $G$ is a cycle. \end{lem}
\begin{proof}
\begin{figure}[H]
\begin{centering}
\includegraphics[scale=0.5]{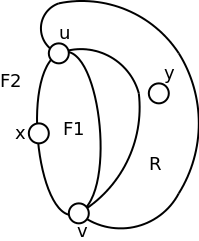}
\par\end{centering}

\caption{\label{fig5:Two-distinct-vertices}Two distinct vertices can lie on
atmost 1 face in a 3-connected planar graph}

\end{figure}
(Refer to Figure \ref{fig5:Two-distinct-vertices}) We prove by contradiction.
Assume there exists such a pair of vertices $(u,v)$. We show that
$(u,v)$ forms a separating pair. Consider the face $F_{1}$, the
outer face $F_{2}$ and the region $R$ representing the remaining
graph. Consider a vertex $x$ on a path from $u$ to $v$ and another
vertex $y$ in $R$. Since the graph is 3-connected, there exist 3
distinct paths between $x$ and $y$, by Menger's theorem. Since $x$
lies on $F_{1}$ and $F_{2}$, the path from $x$ to $y$ must pass
through atleast one of these faces, by Jordan's Curve theorem. But
that will split faces $F_{1}$ or $F_{2}$ into two distinct regions.
As such, any path from $x$ to $y$ must pass through $u$ or $v$,
making $(u,v)$ a separating pair. The other case is when the region
$R$ is empty, which makes G a cycle.
\end{proof}
As a corrollary to the above theorem, we get:
\begin{cor}
In a 3-connected planar graph $G$, two distinct vertices when connected
by an edge splits one face into two new faces, creating exactly $1$
new face.
\end{cor}
Now, due to the above theorem, we can update the $Emb$ and the $Face$
relations.

\subsubsection{$insert(a,b)$}

Any edge $\{a,b\}$ that is inserted lies on a particular face, say
$f$. Consider the edges from vertex $a$. Since $a$ lies on the
face $f$, exactly two edges from $a$ will lie on the boundary of
$f$. Let these two edges, considered anti-clockwise be $e_{1}$ and
$e_{2}$, having the embedded numbering $n_{1}$ and $n_{2}$, respectively.
Note that $n_{2}=(n_{1}+1)$ mod $d_{a}$, where $d_{a}$ is the degree
of $a$. This is because if this was not so, there would be some other
edge in the anti-clockwise direction between $e_{1}$ and $e_{2}$,
which would mean either we have selected a wrong face or the wrong
edges $e_{1}$ and $e_{2}$.

Hence, when we insert the new edge $\{a,b\}$, we can give $\{a,b\}$
the embedded number $n_{2}$, and all the other edges around $a$
which have an embedded number more than $n_{2}$, can be incremented
by $1$. Similarly we do this for $b$. \\

\label{emp-emb-insert} The queries in a high-level form during insertion
are given in \ref{alg-emb,face-insert}. An illustration of the queries
is given in Figure \ref{figB:Splitting-(during-insertion)}.

\subsubsection{$delete(a,b)$}

Since we are in state \emph{A} as mentioned in the Preliminaries part
IV, we expect the graph to be 3-connected and planar once the edge
$(a,b)$ is removed. Hence by the converse of Lemma \emph{\ref{lem-one-face-for-2-distict-vertices}}
above, exactly two faces will get merged. As such, our queries now
will be the exact opposite to those for insertion. \\

\label{emp-emb-delete} The queries in a high-level form during deletion
are given in \ref{alg-emb,face-delete}.

\subsubsection{Rotating and flipping the Embedding}

We will show how to rotate or flip the embedding of the graph in \textsf{FOL}
if required, as it will be necessary for further sections.

The type of rotation that we will accomplish in this section is as
follows: In any given \textsc{CBFS} tree $[v,v_{e}]$, for every vertex
$x$, we rotate the embedding around $x$ until its parent gets the
least embedding number, number $0$ (that is the $0$'th number in
the ordering). For the root vertex $v$ which has no parent, we give
$v_{e}$ the least embedding number.

This scheme is like a\emph{'normal'} form for ordering the edges around
any vertex, or\emph{ 'normalizing'} the embedding. We show in this
section that this can be done in \textsf{FOL}. Also, flipping the
ordering from anti-clockwise to clockwise is (very easily) in \textsf{FOL}.

We shall create the following relation: $Emb_{p}(v,v_{e},t,x,n_{x})$,
which will mean that in the \textsc{CBFS} tree $[v,v_{e}]$, for some
vertex $t$, if the parent of $t$ is $t_{p}$, and if the edge $(t,t_{p})$
(or the vertex $t_{p}$) is given the embedded number $0$, then the
edge $(t,x)$ (or the vertex $x$) gets the embedded number $n_{x}$.

Note that our relation $Emb$ was independent of any particular \textsc{CBFS}
tree, since it depended only on the structure of the 3-connected planar
graph and not on any \textsc{CBFS} tree we chose. But $Emb_{p}$ depends
on the chosen \textsc{CBFS} tree. Another thing to note is that we
do not maintain the relation $Emb_{p}$ in our vocabulary $\tau$,
since it can be easily created \textsf{in FOL} from the rest of the
relations whenever required.

We create the relation $Emb_{p}$ in the following manner. In every
\textsc{CBFS} tree $[v,v_{e}]$, for every vertex $t$, we find the
degree ($d_{t}$) and the parent ($t_{p}$) of $t$, and the embedded
number $n_{p}$ of $t_{p}$. Then for every vertex $x$ in the neighbourhood
of $t$ with embedded numbering $n_{x}$, we do $n_{x}=(n_{x}-n_{p})$
mod $d_{t}$. Refer to Figure \ref{fig7:Normalized-CBFS-tree} for
an illustration. 

\begin{figure}
\begin{centering}
\includegraphics[scale=0.5]{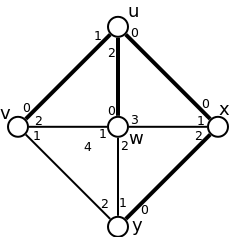}
\par\end{centering}

\caption{\label{fig7:Normalized-CBFS-tree}The CBFS tree of $[u,x]$ with the
normalized rotated embedding}

\end{figure}

We also create the relation $Emb_{f}$ which will contain the flipped
or the clockwise embedding $\pi^{c}$.

\label{emp-emb-rotate} The queries in a high-level form for rotating
and flipping the embedding are given in \ref{alg-emb-rotate-flip}.

A note said throughout in the manuscript is necessary to repeat here.
Though the parent of $v$ is null in $[v,v_{e}]$, we allow the parent
of $v$ to be $v_{e}$, so as to keep the queries neater. If this
convention is not required, then the special case of the parent of
$v$ can be handled easily by modifying the queries.

This shows that\emph{ }the embedding can be \emph{flipped} and \emph{normalized}
in \textsf{FOL}. We conclude the following:
\begin{thm}
The embedding of a 3-connected planar graph can be maintained, normalized
and flipped in $\textsf{\textsf{\textsf{\textsf{DynFO}}}}$. 
\end{thm}

\subsection{Maintaining $CBFSEdges(v,v_{e},s,t)$ and $CPath(v,v_{e},x,y,z)$}

In this section, we show how to maintain the final two relations via
insertions and deletions of tuples that will help us to decide the
isomorphism of two graphs. The relations are almost completely similar
to the ones used for Breadth-First Search in the previous section.
The only difference which arises is due to the uniqueness of the paths
in Canonical Breadth-First Search Trees. We do not rewrite the $Level(v,x,l)$
since it will be exactly similar to the general BFS case.

\subsubsection{$insert(a,b)$}

The \textsc{CBFS} tree is unique if the path to every vertex $x$
from the root vertex $v$ is uniquely defined. How shall we choose
the unique path? First, we consider the paths with the shortest length.
This is exactly same as in Breadth-First Search seen in the previous
section. But unlike \textsc{BFS}, where we chose the shortest path
arbitrarily (that is by lexicographic ordering during insertion/deletion),
we will very precisely choose one of the paths from the set of shortest
paths, by Definition \ref{def-can-order}. Intuitively, Definition
\ref{def-can-order} chooses the path based on its orientation according
to the embedding $\pi$.

An important observation from Definition \ref{def-can-order} is the
following: \emph{Distance has preference over Orientation}. This means
that if there are two paths $P_{1}$ and $P_{2}$ from $v$ to $x$
in $[v,v_{e}]$ (due to the insertion of an edge which created a cycle
in the tree $[v,v_{e}]$), though $P_{2}<_{c}P_{1}$, the path $P_{1}$
will be chosen if $|P_{1}|<|P_{2}|$ irrespective of the canonical
ordering $<_{c}$.

Consider some $[v,v_{e}]$. During insertion of $\{a,b\}$, let the
old path (from $v$ to some $x$) be $P_{1}$ and assume that the
new path $P_{2}$ passes through $(a,b)$. If $|P_{1}|<|P_{2}|$ or
$|P_{1}|=|P_{2}|\wedge P_{1}<_{c}P_{2}$, the path to $x$ does not
change, and all the edges and tuples to $x$ in the old relations
will belong to the new relations. If $|P_{2}|<|P_{1}|$ or $|P_{1}|=|P_{2}|\wedge P_{2}<_{c}P_{1}$,
the path to $x$ changes. In this case, the new path will be from
$v$ to $a$ in $[v,v_{e}]$, the edge $\{a,b\}$ (from $a$ to $b$),
and the path from $b$ to $x$ in $[b,b_{e}]$. The way we choose
$b_{e}$ is as follows (Refer to Figure \ref{fig6:Choosing-the-vertex}
for an illustration): We find the set of vertices $C$ that are adjacent
to $b$ and are at $level_{v}(b)+1$. Since $a$ will be the parent
of $b$ in $[v,v_{e}]'$, we rotate the embedding around $b$ until
$a$ gets the value $0$, and the choose $b_{e}$ to be the vertex
in $C$ that gets the least embedding number.

\begin{figure}[H]
\begin{centering}
\includegraphics[scale=0.5]{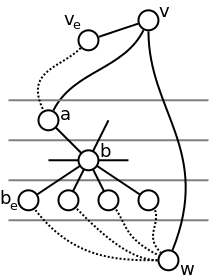}
\par\end{centering}

\caption{\label{fig6:Choosing-the-vertex}Choosing the vertex $b_{e}$ on the
new path from $v$ to $w$}

\end{figure}

To check for the condition $P_{1}<_{c}P_{2}$, we do the following
(Refer to Figure \ref{fig8:Illustrating-can-order}): We create the
set $Emb_{p}$ so that the parent of each vertex has the least embedding
number. Let the path $P_{a}$ denote the path from $v$ to $a$, which
will be a subset of $P_{2}$. We choose the vertex which is the least
common ancestor of $a$ and $x$, say $d=lca_{a,x}$, and normalize
the embedding so that $d_{p}$, the parent of $d$, gets the embedding
number $0$. Existence of $lca_{a,x}$ is guaranteed since $v$ lies
on both $P_{1}$ and $P_{a}$. Now consider the edge $e_{1}=(lca_{a,x},d_{1})$
on $P_{1}$ and $e_{2}=(lca_{a,x},d_{2})$ on $P_{2}$. Since the
embedding is normalized, we see which edge gets the smaller embedding
number around the vertex $lca_{a,x}$. The path on which that edge
lies will be the lesser ordered path according to $<_{c}$. It is
nice to pause here for a moment and observe that this was possible
since the embedding was 'normalized', otherwise it would not have
been possible.

\begin{figure}[H]
\begin{centering}
\includegraphics[scale=0.5]{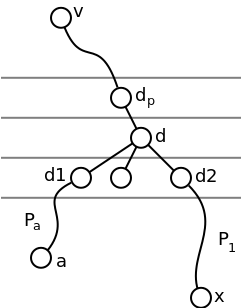}
\par\end{centering}

\caption{\label{fig8:Illustrating-can-order}Illustrating $P_{a}<_{c}P_{1}$
by considering $d=lca_{a,x}$, $d_{1}$, and $d_{2}$ }

\end{figure}

One more thing needs to be shown. In Lemma \ref{lem-level-change-during-insertion},
we proved that for any vertex $x$, its level cannot change both in
the \textsc{BFS} trees of $a$ and $b$. In the previous case of \textsc{BFS},
as per our algorithm, the level not being changed implied the path
not being changed. But that is not the case in \textsc{CBFS} trees.
In \textsc{CBFS} trees, the level not changing may still imply that
the path changes (due to the $<_{c}$ ordering on paths). Hence, it
may be possible that though the level of the vertex $x$ changes in
only one of the \textsc{CBFS} trees, its actual path changes in \emph{both}
the \textsc{CBFS} trees. We need to show that this is not possible.
And the reason this is necessary is because (just like the previous
case) the updation of the path will depend on one specific path to
vertex $x$ in the \textsc{CBFS} tree of $a$ or $b$ which has not
changed.
\begin{lem}
After the insertion of edge $\{a,b\}$, the path to any vertex $x$,
cannot change in both the CBFS trees $[a,a_{e}]$ and $[b,b_{e}]$,
for all $a_{e},b_{e}$. \end{lem}
\begin{proof}
Before the insertion of $\{a,b\}$, let $p_{a}=path_{a,a_{e}}(a,x)$,
and $p_{b,b_{e}}=path_{b,b_{e}}(b,x)$. Without loss of generality,
let $|p_{a}|\leq|p_{b}|$. It suffices to show that after the insertion
of $\{a,b\}$, the path from $a$ to $x$, $p'_{a,a_{e}}(a,x)$, is
the same as $p_{a,a_{e}}(a,x)$. As seen in Lemma \ref{lem-level-change-during-insertion},
$|p{}_{a,a_{e}}(a,x)|=|p'_{a,a_{e}}(a,x)|\leq|p'_{b,b_{e}}(b,x)|$.
If the actual path changes, i.e. $p'_{a,a_{e}}(a,x)\not=p_{a,a_{e}}(a,x)$,
then $|p'_{a,a_{e}}(a,x)|=1+|p'_{a,a_{e}}(b,x)|\geq1+|p'_{b,b_{e}}(b,x)|\geq1+|p_{a,a_{e}}(a,x)|>|p_{a,a_{e}}(a,x)|$
which is a contradiction. 
\end{proof}
\label{emp-cbfs-insert} The queries in a high-level form to maintain
$CBFSEdges$ and $CPath$ during insertion are given in \ref{alg-cbfs-insert}.

\subsubsection{$delete(a,b)$}

For the deletion operation, we choose the edge from $PR_{min}$ based
on the $<_{c}$ relation. Note that when some edge $\{a,b\}$ is deleted,
the path to some vertex $x$ in $[v,v_{e}]$ cannot change if $\{a,b\}$
does not lie on the path. Other things remain exactly similar to the
general case. \label{emp-cbfs-delete}The queries are given in \ref{alg-cbfs-delete}.

\section{Canonization and Isomorphism Testing}

In \cite{thierauf2010isomorphism}, the 3-connected planar graph is
canonized after performing Canonical BFS by using Depth-First Search
(DFS). Performing DFS or any method that employs computing the transitive
closure in any manner cannot be used here since it would not be possible
in $\textsf{FOL}$, and most of the known methods of canonization
seem to require computing the transitive closure. Note that a canon
is required for condition 1 in Definition \ref{def-dynamiccomplexityclass}
to hold. What we seek is a method to canonize the graph, which depends
only on the properties of vertices that can be inferred globally.

To achieve this, we shall label each vertex with a vector. Though
the label will not be succinct now, it will be possible to create
it in $\textsf{FOL}$. 

Essentially\emph{, the canon for a vertex $x$ in some CBFS tree $[v,v_{e}]$
will be a set of tuples $(l,h)$ of the levels and (normalized) embedding
numbers of ancestors of $x$}. 
\begin{defn}
\label{def:canon}Let canon for each vertex $x$ in $[v,v_{e}]$ be
represented by $Canon_{v,v_{e}}(x)$. Then,

\[
{\displaystyle Canon_{v,v_{e}}(x)=\{(l,h):\;\exists q,q_{p},\ C\ \wedge\ L\ \wedge\ P\ \wedge\ H\}}
\]

where,
\begin{itemize}
\item ${\displaystyle C:\ CPath(v,v_{e},v,x,q),}$
\item ${\displaystyle L:\ l=level_{v}(q)}$,
\item ${\displaystyle P:\ parent_{v,v_{e}}(q,q_{p})}$,
\item ${\displaystyle H:\ h=emnum_{v,v_{e}}(q_{p},q)}$\\

\end{itemize}
\end{defn}
\begin{lem}
For any CBFS tree $[v,v_{e}]$, for any two vertices $x$ and $y$,
$x=y\Leftrightarrow Canon(x)=Canon(y)$ \end{lem}
\begin{proof}
If two vertices are same, they have the same canon. If they are different,
it suffices to show that the canons necessarily differ at one point.
Let $d=lca_{v,v_{e}}(x,y)$. Since $d$ is the least common ancestor,
the path to $x$ and $y$ splits at $d$. From Definition \ref{def-can-order},
$(d,d_{x})$ and $(d,d_{y})$ are distinct edges, and they have a
\emph{different} embedding number. Hence the canons for $x$ and $y$
are necessarily different at the level of $d_{x}$ and $d_{y}$.
\end{proof}
It is now easy to canonize each of the CBFS trees in $\textsf{FOL}$.
Once each vertex has a canon, each edge is also uniquely numbered.
The main idea is this: A canon will in itself encode all the necessary
properties of the vertex, and the set of canons of all vertices become
the signature of the graph, preserving edges. The main advantage of
Definition \ref{def:canon} is that the canon of the graph can be
generated in $\textsf{FOL}$. It's worthwhile to observe how this
neatly beats the otherwise inevitable computation of transitive closure
(Theorem \ref{thm-Transitive-Closure-not-in-fol}) to canonize the
graph.

Hence, two 3-connected planar graphs $G$ and $H$ are isomorphic
if and only if for some CBFS tree $[g,g_{e}]$ of $G$, there is a
CBFS tree $[h,h_{e}]$, such that:
\begin{itemize}
\item $\forall x\exists y,\;(x\in G\wedge y\in H\wedge(Canon(x)=Canon(y)))$
and 
\item $\forall x_{1},x_{2},\;((Edge(x_{1},x_{2})\in G)\Leftrightarrow(Edge(Canon(x_{1}),Canon(x_{2}))\in H))$
\end{itemize}
$H$ implies either $H$ with the embedding $\rho$ or $H$ with flipped
embedding $\rho^{-1}$. It is evident that if the graphs are isomorphic,
there will be some CBFS tree in $G$ and $H$ whose canons will be
equivalent in the above sense. If the graphs are not isomorphic, no
canon of any CBFS tree could be equivalent, since it would then directly
give a bijection between the vertices of the graph that preserves
the edges, which would be a contradiction. 

Since we still need to precompute all the relations before the condition
of 3-connectivity is reached, isomorphism of $G$ and $H$ is in $\textsf{\textsf{\textsf{\textsf{DynFO}}}}^{+}$.
This brings us to the main conclusion of this section:
\begin{thm}
3-connected planar graph isomorphism is in $\textsf{\textsf{\textsf{\textsf{DynFO}}}}^{+}$ 
\end{thm}

\section{Conclusions}

We have proven that Breadth-First Search for undirected graphs can be
performed in $\textsf{DynFO}$ and isomorphism for Planar 3-connected
graphs can be decided in $\textsf{\textsf{\textsf{\textsf{DynFO}}}}^{+}$.
A natural extension is to show that Planar Graph Isomorphism is in
$\textsf{DynFO}$. Though even parallel algorithms for this problem
are known \cite{datta2009planar}, the ideas cannot be directly employed
because of myriad problems arising due to automorphisms of the bi/tri-connected
component trees (which are used in \cite{datta2009planar}), and various
subroutines that require computing the transitive closure. In spite
of these shortcomings, we strongly believe that Planar Graph Isomorphism
is in $\textsf{\textsf{\textsf{\textsf{DynFO}}}}$, though the exact
nature of the queries still remains open.

\section{Acknowledgements }

The author sincerely thanks Samir Datta for fruitful discussions and
critical comments on all topics ranging from the problem statement
to the preparation of the final manuscript. \\
\\

\bibliographystyle{plain}
\phantomsection\addcontentsline{toc}{section}{\refname}\bibliography{paper}

\smallskip{}

\appendix
\begin{appendices}

\section{Arithmetic}

\subsection{Relations: Sum}

\subsubsection{$insert(a,b)$}

\label{alg-summations-dynfo}

Updating the relation $Sum$ during insertion:
\begin{itemize}
\item A tuple $(t,x,y)$ belongs to $Sum'$ if:\\
\\
\begin{lyxgreyedout}
\{case in which there is no element in the universe\}%
\end{lyxgreyedout}
\\
$t=x=y=a\ \wedge\ \neg U(a)\wedge\forall u(O'(a,u)$\\
OR\\
\\
\begin{lyxgreyedout}
\{adding the relation $a+0=a$\}%
\end{lyxgreyedout}
\\
$min\leftarrow min:\ \forall u,\, U(u)\Rightarrow O(min,u)$\\
$t=a\wedge((x=a\wedge y=min)\vee(x=min\wedge y=a))$\\
OR\\
\\
$max\gets max:$$\forall u,\, U(u)\Rightarrow O(u,max)$\\
\begin{lyxgreyedout}
\{$max$ is the maximum element. Note that if $a$ is a new element
inserted in the universe, $max+1=a$\}%
\end{lyxgreyedout}
\\
$Sum(t,x,y)\vee(t=a\wedge(Sum(max,x-1,y)\vee Sum(max,x,y-1)))$\\

\end{itemize}
(Back to Section \ref{to-summation-dynfo})

\section{Breadth-First Search}

\subsection{Relations: Level, BFSEdge, BFSPath}

\subsubsection{$insert(a,b)$ }

\label{alg-bfs-insert}

The relations $Level$, $BFSEdge$, and $BFSPath$ are updated during
insertion a follows: \{The queries are illustrated in Figure \ref{figA:-bfs-insert}\}
\begin{itemize}
\item A tuple $(v,x,l)$ belongs to $Level'$, if:\\
${\displaystyle l_{old}\leftarrow level_{v}(x)}$ %
\begin{lyxgreyedout}
\{$l_{old}\leftarrow\infty$ if there is no $l_{old}$ such that $Level(v,x,l_{old})$
holds\}%
\end{lyxgreyedout}
\\
IF $level_{a}(x)\leq level_{b}(x)$: $\alpha\leftarrow b,$ $\beta\leftarrow a$;
ELSE: $\alpha\leftarrow a,$ $\beta\leftarrow b$\\
$l_{new}\leftarrow level_{v}(\alpha)+1+level_{\beta}(x)$\\
$l=\min(l_{old},l_{new})$\\

\item A tuple $(v,x,y)$ belongs to $BFSEdge'$, if:\\
$\exists w$ such that \\
$l_{old}\leftarrow level_{v}(w)$, $l_{new}\leftarrow level'_{v}(w)$,
\\
IF $level_{a}(w)\leq level_{b}(w)$: $\alpha\leftarrow b,$ $\beta\leftarrow a$;
ELSE: $\alpha\leftarrow a,$ $\beta\leftarrow b$\\
\\
\begin{lyxgreyedout}
\{the level of $w$ did not change and $\{x,y\}$ was on a path from
$v$ to $w$\}%
\end{lyxgreyedout}
\\
$l_{new}=l_{old}$ and $Path(v,v,w,\{x,y\})\wedge BFSEdge(v,x,y)$\\
\textcolor{black}{\emph{OR}}\\
\\
\begin{lyxgreyedout}
\{the level of $w$ changed and $\{x,y\}$ lies on the new path\}%
\end{lyxgreyedout}
\\
$l_{new}<l_{old}$ and \\
$(Path(v,v,\alpha,\{x,y\})\wedge BFSEdge(v,x,y))$ %
\begin{lyxgreyedout}
\{Path from $v$ to $\alpha$\}%
\end{lyxgreyedout}
\\
$\vee(Path(\beta,\beta,w,\{x,y\})\wedge BFSEdge(\beta,x,y))$ %
\begin{lyxgreyedout}
\{Path from $\beta$ to $w$\}%
\end{lyxgreyedout}
\\
$\vee(x=a\wedge y=b)$$\vee(x=b\wedge y=a)$ %
\begin{lyxgreyedout}
\{$\{x,y\}$ is the edge $\{a,b\}$\}%
\end{lyxgreyedout}
\\

\item A tuple $(v,x,y,z)$ belongs to $Path'$, if:\\
$\exists w$ such that \\
$l_{old}\leftarrow level_{v}(w)$, $l_{new}\leftarrow level'_{v}(w)$,
\\
IF $level_{a}(w)\leq level_{b}(w)$: $\alpha\leftarrow b,$ $\beta\leftarrow a$;
ELSE: $\alpha\leftarrow a,$ $\beta\leftarrow b$\\
\\
\begin{lyxgreyedout}
\{the level of $w$ did not change and $\{x,y\}$ was on a path from
$v$ to $w$\}%
\end{lyxgreyedout}
\\
$l_{new}=l_{old}$ and $Path(v,v,w,\{x,y,z\})\wedge Path(v,x,y,z)$\\
OR\\
\\
\begin{lyxgreyedout}
\{the level of $w$ changed, and the vertices $x,y,z$ lie on the
new path\}%
\end{lyxgreyedout}
\\
$l_{new}<l_{old}$ and \\
$(Path(v,v,\alpha,\{x,y,z\})\wedge Path(v,x,y,z))$ %
\begin{lyxgreyedout}
\{All on the path from $v$ to $\alpha$\}%
\end{lyxgreyedout}
\\
$\vee(Path(\beta,\beta,w,\{x,y,z\})\wedge Path(\beta,x,y,z))$ %
\begin{lyxgreyedout}
\{All on the path from $\beta$ to $w$\}%
\end{lyxgreyedout}
\\
$\vee(Path(v,v,\alpha,\{x\})\wedge Path(\beta,\beta,w,\{y,z\})\wedge Path(\beta,\beta,y,z))$
\begin{lyxgreyedout}
\{$x$ on $path_{v}(v,\alpha)$ and $y,z$ on $path_{\beta}(\beta,w)$\}%
\end{lyxgreyedout}
\\
$\vee(Path(v,v,\alpha,\{x,z\})\wedge Path(v,v,z,x)\wedge Path(\beta,\beta,w,y))$
\begin{lyxgreyedout}
\{$x,z$ on $path_{v}(v,\alpha)$ and $y$ on $path_{\beta}(\beta,w)$\}%
\end{lyxgreyedout}
\\
\\
\begin{figure}[h]
\begin{centering}
\includegraphics[scale=0.5]{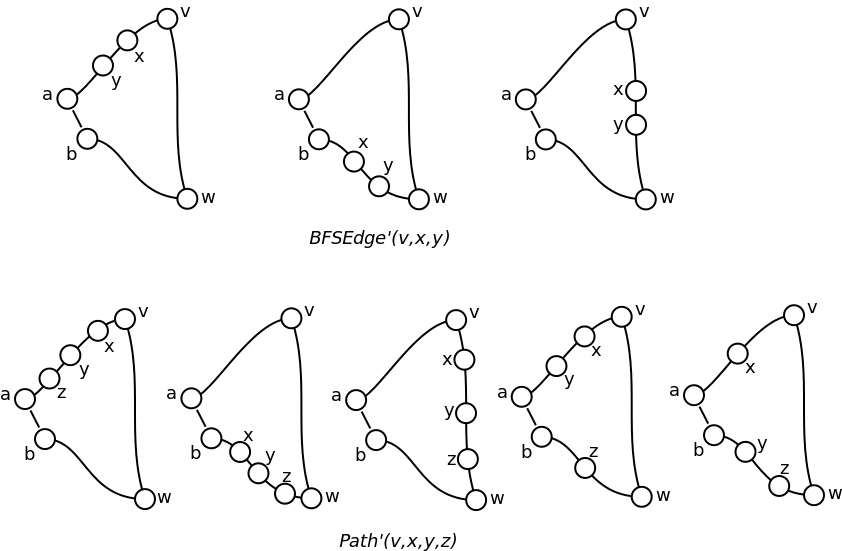}
\par\end{centering}

\caption{\label{figA:-bfs-insert}$BFSEdge$ and $Path$ during insertion of
edge $\{a,b\}$}
\end{figure}

\end{itemize}
(Back to Section \ref{emp-bfs-insert})

\subsubsection{$delete(a,b)$}

\label{alg-bfs-delete}

The relations $Level$, $BFSEdge$, and $BFSPath$ are updated during
deletion are as follows (Refer to Figure \ref{fig3:Deletion-of-edge-bfs}
for illustration of relations used):
\begin{itemize}
\item $R_{2}(v,x)=BFSEdge(v,a,b)\wedge Path(v,v,x,\{a,b\})$\\
$R_{1}(v,y)=\neg R_{2}(v,y)$\\
$PR(v,s,t)=R_{1}(v,s)\wedge R_{2}(v,t)\wedge Edge(s,t)$ %
\begin{lyxgreyedout}
\{All edges connecting $R_{1}$ and $R_{2}$\}%
\end{lyxgreyedout}
\\
$l_{min}(v,w)\leftarrow\min\{level_{v}(s)+1+level_{t}(w):\; PR(v,s,t)\}$
\begin{lyxgreyedout}
\{Length of the new shortest path from $v$ to $w$\}%
\end{lyxgreyedout}
\\
$PR_{min}(v,w,s,t)=R_{2}(v,w)\wedge PR(v,s,t)\wedge(level_{v}(s)+1+level_{t}(w)=l_{min}(v,w))$
\begin{lyxgreyedout}
\{Set of edges that lead to the shortest path\}%
\end{lyxgreyedout}
\\
$PR_{lex,min}(v,w,s,t)=PR_{min}(v,w,s,t)\wedge(s\leq t)\wedge(\forall p,q,\; PR_{min}(v,w,p,q)\Rightarrow(s<p)\vee((s=p)\wedge(t\leq q)))$
\\
\begin{lyxgreyedout}
\{Choosing the lexicographically smallest edge. $PR_{lex,min}$ is
the set of new edges that will be added. The queries are now exactly
similar to insertion of edges\}%
\end{lyxgreyedout}

\item A tuple $(v,x,l)$ belongs to $Level'$ if:\\
\\
\begin{lyxgreyedout}
\{$\{a,b\}$ did not belong to $v$'s BFS tree\}%
\end{lyxgreyedout}
\\
$\neg BFSEdge(v,a,b)\wedge Level(v,x,l)$\\
OR\\
$BFSEdge(v,a,b)\wedge l=l_{min}(v,x)$\\

\item A tuple $(v,x,y)$ belongs to $BFSEdge'$ if:\\
$\exists w$ such that\\
\\
\begin{lyxgreyedout}
\{$\{a,b\}$ was not and $\{x,y\}$ was on the path from $v$ to $w$\}%
\end{lyxgreyedout}
\\
$\neg Path(v,v,w,\{a,b\})\wedge Path(v,v,w,\{x,y\})\wedge BFSEdge(v,x,y)$\\
OR\\
\\
\begin{lyxgreyedout}
\{$\{x,y\}$ lies on the new path from $v$ to $w$\}%
\end{lyxgreyedout}
\\
$p,r\leftarrow PR_{lex,min}(v,w,p,r)$\\
$(Path(v,v,w,\{a,b\})\wedge Path(v,v,p,\{x,y\})\wedge BFSEdge(v,x,y))$
\begin{lyxgreyedout}
\{edge $\{x,y\}$ is on $path_{v}(v,p)$\}%
\end{lyxgreyedout}
\\
$\vee(Path(r,r,w,\{x,y\})\wedge BFSEdge(r,x,y))$ %
\begin{lyxgreyedout}
\{edge $\{x,y\}$ is on $path_{r}(r,w)$\}%
\end{lyxgreyedout}
\\
$\vee((x=p)\wedge(y=r))\vee((x=r)\wedge(y=p))$ %
\begin{lyxgreyedout}
\{$\{x,y\}$ is the edge $\{p,r\}$\}%
\end{lyxgreyedout}
\\

\item A tuple $(v,x,y,z)$ belongs to $Path'$ if:\\
$\exists w$ such that\\
\\
\begin{lyxgreyedout}
\{path from $v$ to $w$ is unchanged, and $z$ is on path from $x$
to $y$ in BFS tree of $v$\}%
\end{lyxgreyedout}
\\
$\neg Path(v,v,w,\{a,b\})\wedge Path(v,v,w,\{x,y,z\})\wedge Path(v,x,y,z)$\\
OR\\
\\
\begin{lyxgreyedout}
\{$\{x,y,z\}$ lies on the new path from $v$ to $w$\}%
\end{lyxgreyedout}
\\
$p,r\leftarrow PR_{lex,min}(v,w,p,r)$\\
$(Path(v,v,w,\{a,b\})\wedge Path(v,v,p,\{x,y,z\})\wedge Path(v,x,y,z))$
\{All of $x,y,z$ on the path from $v$ to $p$\}\\
$\vee(Path(r,r,w,\{x,y,z\})\wedge Path(r,x,y,z))$ \{All of $x,y,z$
on the path from $r$ to $w$\}\\
$\vee(Path(v,v,p,\{x\})\wedge Path(r,r,w,\{y,z\})\wedge Path(r,r,y,z))$
\{$x$ on $path_{v}(v,p)$ and $y,z$ on $path_{r}(r,w)$\}\\
$\vee(Path(v,v,p,\{x,z\})\wedge Path(v,v,z,x)\wedge Path(r,r,w,y))$
\{$x,z$ on $path_{v}(v,p)$ and $y$ on $path_{r}(r,w)$\}\\

\end{itemize}
(Back to Section \ref{emp-bfs-delete})

\section{Canonical Breadth-First Search}

\subsection{Canonical Breadth-First Search Method}

\label{alg-cbfs-method}

\begin{algorithm}[H]
$Queue\leftarrow null$;

$Enqueue(v)$;

Add $(v,v_{e})$ to the \textsc{CBFS} tree;

WHILE $!Queue.empty()$

~~~~~$u\leftarrow Dequeue()$; 

~~~~~$u_{p}\leftarrow u.parent()$; 

\begin{lyxgreyedout}
~~~~~\{where $v.parent()$ is let to be $v_{e}$ for ease of
code, though $v$'s parent is actually $null$\} %
\end{lyxgreyedout}

~~~~~$k\leftarrow\pi_{u}(u_{p})$; 

~~~~~$k\leftarrow(k+1)$ mod $d_{u}$; 

~~~~~$u'\leftarrow\pi_{u}^{-1}(k)$; 

~~~~~WHILE $u'\not=u_{p}$ 

~~~~~~~~~~IF $u'$ is not visited

~~~~~~~~~~~~~~~Add $(u,u')$ to the \textsc{CBFS} tree; 

~~~~~~~~~~~~~~~Mark $u'$ as visited;

~~~~~~~~~~$Enqueue(u')$;

~~~~~~~~~~$k\leftarrow(k+1)$ mod $d_{u}$; 

~~~~~~~~~~$u'\leftarrow\pi_{u}^{-1}(k)$; 

\caption{Canonical Breadth-First-Search Method for $(v,v_{e})$}
\end{algorithm}

(Back to Section \ref{emp-cbfs-method})

\subsection{Relations: Emb, Face}

\subsubsection{$insert(a,b)$ }

\label{alg-emb,face-insert}

Updating the relations $Emb$ and $Face$ during insertion. \{Refer
to Figure \ref{figB:Splitting-(during-insertion)} for an illustration
and intuition for the following queries\}\\

\begin{itemize}
\item $f_{ab}\leftarrow Face(f_{ab},a,b,b)$ %
\begin{lyxgreyedout}
\{Face on which both $a$ and $b$ lie\}%
\end{lyxgreyedout}
\\
$a_{2}\leftarrow Edge(a,a_{2})\wedge(\forall z,Face(f_{ab},a_{2},a,z)\Rightarrow z=a\vee z=a_{2})$\\
$b_{2}\leftarrow Edge(b,b_{2})\wedge(\forall z,Face(f_{ab},b_{2},b,z)\Rightarrow z=b\vee z=b_{2})$\\
$n_{a_{2}}\leftarrow Emb(a,a_{2},n_{a_{2}})$\\
$n_{b_{2}}\leftarrow Emb(b,b_{2},n_{b_{2}})$\\

\item A tuple $(v,x,n_{x})$ belongs to $Emb'$ if:\\
\\
\begin{lyxgreyedout}
\{$v$ is not affected by the insertion of $\{a,b\}$\}%
\end{lyxgreyedout}
\\
$(v\not=a)\wedge(v\not=b)\wedge Emb(v,x,n_{x})$\\
OR\\
\\
\begin{lyxgreyedout}
\{the tuple has the new embedding number given by $a$ to $b$ or
$b$ to $a$\}%
\end{lyxgreyedout}
\\
$(v,x,n_{x})=(a,b,n_{a_{2}})\vee(b,a,n_{b_{2}})$\\
OR\\
\\
\begin{lyxgreyedout}
\{the tuple represents a vertex around $a$ or $b$ whose embedding
number has not changed\}%
\end{lyxgreyedout}
\\
$((v=a)\wedge Emb(v,x,n_{x})\wedge(n_{x}<n_{a_{2}}))$\\
$\vee((v=b)\wedge Emb(v,x,n_{x})\wedge(n_{x}<n_{b_{2}}))$\\
OR\\
\\
\begin{lyxgreyedout}
\{the tuple represents a vertex around $a$ or $b$ whose embedding
number has increased by 1\}%
\end{lyxgreyedout}
\\
$((v=a)\wedge Emb(v,x,n_{x}-1)\wedge(n_{x}\geq n_{a_{2}}))$\\
$\vee((v=b)\wedge Emb(v,x,n_{x}-1)\wedge(n_{x}\geq n_{b_{2}}))$\\

\item A tuple $(f,x,y,z)$ belongs to $Face'$ if:\\
\\
\begin{lyxgreyedout}
\{the face was not the one on which both $a$ and $b$ were there\}%
\end{lyxgreyedout}
\\
$f\not=f_{ab}\wedge Face(f,x,y,z)$\\
OR\\
\\
\begin{lyxgreyedout}
\{the face is split into 2 faces\}%
\end{lyxgreyedout}
\\
$S_{1}(z)=Face(f_{ab},a,b,z)$ %
\begin{lyxgreyedout}
\{Splitting all vertices into 2 sets, for each new face\}%
\end{lyxgreyedout}
\\
$S_{2}(z)=Face(f_{ab},b,a,z)$\\
$f_{p}\leftarrow f_{ab}$ %
\begin{lyxgreyedout}
\{Name for the first face\}%
\end{lyxgreyedout}
\\
$F(u)=$$\forall x,y,z,\neg Face(u,x,y,z)$\\
$f_{q}\leftarrow F(f_{q})\wedge\forall x,F_{q}(x)\Rightarrow f_{q}\leq x$
\begin{lyxgreyedout}
\{choosing the lexicographically smallest available label\}%
\end{lyxgreyedout}
\\
$(f=f_{p}\wedge S_{1}(\{x,y,z\}))\vee(f=f_{q}\wedge S_{2}(\{x,y,z\}))$\\
\begin{figure}[H]
\begin{centering}
\includegraphics[scale=0.5]{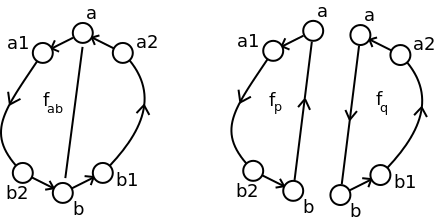}
\par\end{centering}

\caption{\label{figB:Splitting-(during-insertion)}Splitting (during insertion)
and Merging (during deletion) of face(s)}
\end{figure}

\end{itemize}
(Back to Section \ref{emp-emb-insert})

\subsubsection{$delete(a,b)$}

\label{alg-emb,face-delete}

Updating the relations $Emb$ and $Face$ during deletion. \{Refer
to Figure \ref{figB:Splitting-(during-insertion)} for an illustration
and intuition for the following queries\}\\

\begin{itemize}
\item $f_{p}=Face(f_{p},a,b,b)\wedge(\forall x,Face(f_{p},b,a,x)\Rightarrow x=a\vee x=b)$\\
$f_{q}=Face(f_{q},a,b,b)\wedge(\forall x,Face(f_{q},a,b,x)\Rightarrow x=a\vee x=b)$
\begin{lyxgreyedout}
\{Finding the two faces on which the edge $\{a,b\}$ lies\}%
\end{lyxgreyedout}
\\
$f_{ab}=\min(f_{p},f_{q})$ %
\begin{lyxgreyedout}
\{Name for the new combined face\}%
\end{lyxgreyedout}
\\
$n_{a}\leftarrow Emb(b,a,n_{a})$\\
$n_{b}\leftarrow Emb(a,b,n_{b})$\\

\item A tuple $(v,x,n_{x})$ belongs to $Emb'$ if:\\
\\
\begin{lyxgreyedout}
\{$v$ is not affected by the deletion of $\{a,b\}$\}%
\end{lyxgreyedout}
\\
$(v\not=a)\wedge(v\not=b)\wedge Emb(v,x,n_{x})$\\
OR\\
\\
\begin{lyxgreyedout}
\{the tuple represents a vertex around $a$ or $b$ whose embedding
number has not changed\}%
\end{lyxgreyedout}
\\
$((v=a)\wedge Emb(v,x,n_{x})\wedge(n_{x}<n_{b}))$\\
$\vee((v=b)\wedge Emb(v,x,n_{x})\wedge(n_{x}<n_{a}))$\\
OR\\
\\
\begin{lyxgreyedout}
\{the tuple represents a vertex around $a$ or $b$ whose embedding
number has decreased by 1\}%
\end{lyxgreyedout}
\\
$((v=a)\wedge Emb(v,x,n_{x}+1)\wedge(n_{x}>n_{b}))$\\
$\vee((v=b)\wedge Emb(v,x,n_{x}+1)\wedge(n_{x}>n_{a}))$\\

\item A tuple $(f,x,y,z)$ belongs to $Face'$ if:\\
\\
\begin{lyxgreyedout}
\{every tuple in the old relation\}%
\end{lyxgreyedout}
\\
$(f\not=f_{p}\wedge f\not=f_{q}\wedge Face(f,x,y,z))$\\
$\vee(f=f_{ab}\wedge(Face(f_{p},x,y,z)\vee Face(f_{q},x,y,z)))$ \{transferring
all tuples of $f_{p}$ and $f_{q}$ to $f_{ab}$\}\\
OR\\
\\
\begin{lyxgreyedout}
\{the tuple forms paths between vertices from one face to those in
the other\}%
\end{lyxgreyedout}
\\
$(f=f_{ab}\wedge Face(f_{p},a,x,x)\wedge Face(f_{q},b,y,z))$ \{$x$
in $f_{p}$, $y,z$ in $f_{q}$\}\\
$\vee(f=f_{ab}\wedge Face(f_{q},b,x,x)\wedge Face(f_{p},a,y,z))$
\{$x$ in $f_{q}$, $y,z$ in $f_{p}$\}\\
$\vee(f=f_{ab}\wedge Face(f_{p},a,z,x)\wedge Face(f_{q},b,y,y))$
\{$x,z$ in $f_{p}$, $y$ in $f_{q}$\}\\
$\vee(f=f_{ab}\wedge Face(f_{q},b,z,x)\wedge Face(f_{p},a,y,y))$
\{$x,z$ in $f_{q}$, $y$ in $f_{p}$\}\\
$\vee(f=f_{ab}\wedge Face(f_{p},a,z,x)\wedge Face(f_{q},b,y,y))$
\{$x,z$ in $f_{p}$, $y$ in $f_{q}$\}\\
$\vee(f=f_{ab}\wedge Face(f_{p},a,x,x)\wedge Face(f_{q},b,z,z)\wedge Face(f_{p},x,y,b))$
\{$x$ in $f_{p}$, $z$ in $f_{q}$, $y$ in $f_{p}$\}\\
$\vee(f=f_{ab}\wedge Face(f_{q},b,x,x)\wedge Face(f_{p},a,z,z)\wedge Face(f_{q},x,y,a))$
\{$x$ in $f_{q}$, $z$ in $f_{p}$, $y$ in $f_{q}$\}\\

\end{itemize}
(Back to Section \ref{emp-emb-delete})

\subsubsection{Rotating and Flipping the Embedding}

\label{alg-emb-rotate-flip}

Queries to rotate and flip the embedding in $\textsf{FOL}$.
\begin{itemize}
\item %
\begin{lyxgreyedout}
\{$Deg(v,d_{v})$ holds if the degree of vertex $v$ is $d_{v}$\}%
\end{lyxgreyedout}
\\
$Deg(v,d_{v})=(\forall u,Edge(v,u)\Rightarrow\exists n_{u},Emb(v,u,n_{u})\wedge(n_{u}<d_{v}))\wedge\exists u,n_{u},Edge(v,u)\wedge Emb(v,u,n_{u})\wedge(n_{u}+1=d_{v})$\\

\item %
\begin{lyxgreyedout}
\{$Parent(v,v_{e},x_{p},x)$ denotes that in $[v,v_{e}]$, vertex
$x_{p}$ is the parent of vertex $x$ \}%
\end{lyxgreyedout}
\\
$Parent(v,v_{e},x_{p},x)=\exists l_{p},l,Level(v,x_{p},l_{p})\wedge Level(v,x,l)\wedge(l_{p}+1=l)\wedge CBFSEdges(v,v_{e},x_{p},x)$\\

\item %
\begin{lyxgreyedout}
\{$EmbPar(v,v_{e},x,n_{p})$ denotes that the embedding number of
$x$'s parent in $[v,v_{e}]$ is $n_{p}$\}%
\end{lyxgreyedout}
\\
$EmbPar(v,v_{e},x,n_{p})=\exists x_{p},Parent(v,v_{e},x_{p},x)\wedge Emb(x,x_{p},n_{p})$\\
$Emb_{p}(v,v_{e},t,x,n{}_{x})=Edge(x,t)\wedge\exists n_{p},d_{t},n_{old},Deg(t,d_{t})\wedge EmbPar(v,v_{e},t,n_{p})\wedge Emb(t,x,n_{old})$\\
$\wedge(n_{old}\geq n_{p}\Rightarrow n{}_{x}=n_{old}-n_{p})\wedge(n_{old}<n_{p}\Rightarrow n{}_{x}=n_{old}+d_{x}-n_{p})$\\
$Emb_{f}(v,x,n{}_{x})=\exists n_{old},d_{v},Emb(v,x,n_{old})\wedge Deg(v,d_{v})\wedge(n{}_{x}=d_{v}-1-n_{old})$\\
\\
(Back to Section \ref{emp-emb-rotate})
\end{itemize}

\subsection{Relations: CBFSEdges, CBFSPath}

\subsubsection{$insert(a,b)$}

\label{alg-cbfs-insert}

Updating the relations $CBFSEdges$ and $CBFSPath$ during insertion. 
\begin{itemize}
\item $l_{old}\leftarrow level_{v}(w)$, $l_{new}\leftarrow level{}_{v}(a)+1+level_{b}(w)$,
IF $level_{a}(w)\leq level_{b}(w)$: $\alpha\leftarrow b,$ $\beta\leftarrow a$;
ELSE: $\alpha\leftarrow a,$ $\beta\leftarrow b$\\
$d\leftarrow lca_{v,v_{e}}(w,a)$\\
$d_{1}\leftarrow CPath(v,v_{e},d,w,d_{1})\wedge CBFSEdges(v,v_{e},d,d_{1})$
\\
$d_{2}\leftarrow CPath(v,v_{e},d,a,d_{2})\wedge CBFSEdges(v,v_{e},d,d_{2})$\\
$n_{1}\leftarrow emnum'_{v,v_{e}}(d,d_{1})$\\
$n_{2}\leftarrow emnum'_{v,v_{e}}(d,d_{2})$\\
$C(z)=(level'_{v}(z)=level'_{v}(\beta)+1)\wedge Edge(\beta,z)$\\
$\beta_{e}\leftarrow\min\{emnum'_{v,v_{e}}(\beta,z):C(z)\}$\\

\item A tuple $(v,v_{e},x,y)$ belongs to $CBFSEdges'$, if\\
$\exists w$ such that \\
\\
\begin{lyxgreyedout}
\{$|P_{1}|<|P_{2}|$ or $|P_{1}|=|P_{2}|\wedge P_{1}<_{c}P_{2}$,
and $\{x,y\}$ was on $|P_{1}|$\}%
\end{lyxgreyedout}
\\
$(l_{old}<l_{new})\vee(l_{old}=l_{new}\wedge n_{1}<n_{2})$ and $CPath(v,v_{e},v,w,\{x,y\})\wedge CBFSEdges(v,v_{e},x,y)$\\
OR\\
\\
\begin{lyxgreyedout}
\{$|P_{2}|<|P_{1}|$ or $|P_{1}|=|P_{2}|\wedge P_{2}<_{c}P_{1}$,
and $\{x,y\}$ is on $|P_{2}|$\}%
\end{lyxgreyedout}
\\
$(l_{old}>l_{new})\vee(l_{old}=l_{new}\wedge n_{1}>n_{2})$ and \\
$(CPath(v,v_{e},v,\alpha,\{x,y\})\wedge CBFSEdges(v,v_{e},x,y))$
\begin{lyxgreyedout}
\{Path from $v$ to $\alpha$\}%
\end{lyxgreyedout}
\\
$\vee(CPath(\beta,\beta_{e},\beta,w,\{x,y\})\wedge CBFSEdges(\beta,\beta_{e},x,y))$
\begin{lyxgreyedout}
\{Path from $\beta$ to $w$\}%
\end{lyxgreyedout}
\\
$\vee(x=a\wedge y=b)$$\vee(x=b\wedge y=a)$ %
\begin{lyxgreyedout}
\{$\{x,y\}$ is the edge $\{a,b\}$\}%
\end{lyxgreyedout}
\\

\item A tuple $(v,x,y,z)$ belongs to $CPath'$, if\\
$\exists w$ such that \\
\\
\begin{lyxgreyedout}
\{$|P_{1}|<|P_{2}|$ or $|P_{1}|=|P_{2}|\wedge P_{1}<_{c}P_{2}$,
and $\{x,y,z\}$ were on $|P_{1}|$\}%
\end{lyxgreyedout}
\\
$(l_{old}<l_{new})\vee(l_{old}=l_{new}\wedge n_{1}<n_{2})$ and $CPath(v,v_{e},v,w,\{x,y,z\})\wedge CBFSEdges(v,v_{e},x,y,z)$\\
OR\\
\\
\begin{lyxgreyedout}
\{$|P_{2}|<|P_{1}|$ or $|P_{1}|=|P_{2}|\wedge P_{2}<_{c}P_{1}$,
and $\{x,y,z\}$ are on $|P_{2}|$\}%
\end{lyxgreyedout}
\\
$(l_{old}>l_{new})\vee(l_{old}=l_{new}\wedge n_{1}>n_{2})$ and \\
$(CPath(v,v_{e},v,\alpha,\{x,y,z\})\wedge CPath(v,v_{e},x,y,z))$
\begin{lyxgreyedout}
\{All on the path from $v$ to $\alpha$\}%
\end{lyxgreyedout}
\\
$\vee(CPath(\beta,\beta_{e},w,\{x,y,z\})\wedge CPath(\beta,\beta_{e},x,y,z))$
\begin{lyxgreyedout}
\{All on the path from $\beta$ to $w$\}%
\end{lyxgreyedout}
\\
$\vee(CPath(v,v_{e},v,\alpha,\{x\})\wedge CPath(\beta,\beta_{e},\beta,w,\{y,z\})\wedge CPath(\beta,\beta_{e},\beta,y,z))$
\begin{lyxgreyedout}
\{$x$ on $path_{v,v_{e}}(v,\alpha)$ and $y,z$ on $path_{\beta,\beta_{e}}(\beta,w)$\}%
\end{lyxgreyedout}
\\
$\vee(CPath(v,v_{e},v,\alpha,\{x,z\})\wedge CPath(v,v_{e},v,z,x)\wedge CPath(\beta,\beta_{e},\beta,w,y))$
\begin{lyxgreyedout}
\{$x,z$ on $path_{v,v_{e}}(v,\alpha)$ and $y$ on $path_{\beta,\beta_{e}}(\beta,w)$\}%
\end{lyxgreyedout}
\\

\end{itemize}
(Back to Section \ref{emp-cbfs-insert})

\subsubsection{$delete(a,b)$}

\label{alg-cbfs-delete}
\begin{itemize}
\item $R_{2}(v,v_{e},x)=CBFSEdges(v,v_{e},a,b)\wedge CPath(v,v_{e}v,x,\{a,b\})$\\
$R_{1}(v,v_{e},y)=\neg R_{2}(v,v_{e},y)$\\
$PR(v,v_{e}s,t)=R_{1}(v,v_{e},s)\wedge R_{2}(v,v_{e},t)\wedge Edge(s,t)$
\begin{lyxgreyedout}
\{All edges connecting $R_{1}$ and $R_{2}$\}%
\end{lyxgreyedout}
\\
$l_{min}(v,w)\leftarrow\min\{level_{v}(s)+1+level_{t}(w):\;\bigcup_{v_{e}}PR(v,v_{e},s,t)\}$
\begin{lyxgreyedout}
\{Length of the new shortest path from $v$ to $w$\}%
\end{lyxgreyedout}
\\
$PR_{min}(v,v_{e},w,s,t)=R_{2}(v,v_{e},w)\wedge PR(v,v_{e},s,t)\wedge(level_{v}(s)+1+level_{t}(w)=l_{min}(v,w))$
\begin{lyxgreyedout}
\{Set of edges that lead to the shortest path\}%
\end{lyxgreyedout}
\\
$PR_{<,min}(v,v_{e},w,s,t)=PR_{min}(v,v_{e},w,s,t)\wedge$\\
$(\forall p,q,\; PR_{min}(v,w,p,q)\Rightarrow(path_{v,v_{e}}(v,s)<_{c}path_{v,v_{e}}(v,p))\vee(s=p\wedge emnum_{v,v_{e}}(s,t)\leq emnum_{v,v_{e}}(s,q)))$\\
\begin{lyxgreyedout}
\{$PR_{<,min}$ is the set of new edges that will be added. The queries
are now similar to insertion of edges\}%
\end{lyxgreyedout}
\\

\item A tuple $(v,v_{e},x,y)$ belongs to $CBFSEdges'$, if\\
$\exists w$ such that \\
\\
\begin{lyxgreyedout}
\{$\{a,b\}$ was not on $path_{v,v_{e}}(v,w)$, and $\{x,y\}$ was
on $path_{v,v_{e}}(v,w)$\}%
\end{lyxgreyedout}
\\
$\neg CPath(v,v_{e},v,w,\{a,b\})\wedge CPath(v,v,w,\{s,t\})\wedge CBFSEdges(v,s,t)$\\
OR\\
\\
\begin{lyxgreyedout}
\{$\{x,y\}$ lies on the new path from $v$ to $w$\}%
\end{lyxgreyedout}
\\
$s,t\leftarrow PR_{<,min}(v,v_{e},w,s,t)$\\
$C(z)=(level'_{v}(z)=level'_{v}(t)+1)\wedge Edge(t,z)$\\
$t_{e}\leftarrow\min\{emnum'_{v,v_{e}}(t,z):C(z)\}$\\
$(CPath(v,v_{e},v,w,\{a,b\})\wedge CPath(v,v_{e},v,s,\{x,y\})\wedge CBFSEdges(v,v_{e},x,y))$
\begin{lyxgreyedout}
\{edge $\{x,y\}$ is on $path_{v,v_{e}}(v,s)$\}%
\end{lyxgreyedout}
\\
$\vee(CPath(t,t_{e},t,w,\{s,t\})\wedge CBFSEdges(t,t_{e},x,y))$ %
\begin{lyxgreyedout}
\{edge $\{x,y\}$ is on $path_{t,t_{e}}(t,w)$\}%
\end{lyxgreyedout}
\\
$\vee((x=s)\wedge(y=t))\vee((x=t)\wedge(y=s))$ %
\begin{lyxgreyedout}
\{$\{x,y\}$ is the edge $\{s,t\}$\}%
\end{lyxgreyedout}
\\

\item A tuple $(v,v_{e},x,y,z)$ belongs to $CPath'$, if\\
$\exists w$ such that\\
\\
\begin{lyxgreyedout}
\{path from $v$ to $w$ is unchanged, and $z$ is on path from $x$
to $y$ in $[v,v_{e}]$\}%
\end{lyxgreyedout}
\\
$\neg CPath(v,v_{e},v,w,\{a,b\})\wedge CPath(v,v_{e},v,w,\{x,y,z\})\wedge CPath(v,v_{e},x,y,z)$\\
OR\\
\\
\begin{lyxgreyedout}
\{$\{x,y,z\}$ lies on the new path from $v$ to $w$\}%
\end{lyxgreyedout}
\\
$s,t\leftarrow PR_{<,min}(v,w,s,t)$\\
$C(z)=(level'_{v}(z)=level'_{v}(t)+1)\wedge Edge(t,z)$\\
$t_{e}\leftarrow\min\{emnum'_{v,v_{e}}(t,z):C(z)\}$\\
$(CPath(v,v_{e},v,w,\{a,b\})\wedge CPath(v,v_{e},v,s,\{x,y,z\})\wedge CPath(v,v_{e},x,y,z))$
\begin{lyxgreyedout}
\{All of $x,y,z$ on $path_{v,v_{e}}(v,s)$\}%
\end{lyxgreyedout}
\\
$\vee(CPath(t,t_{e},t,w,\{x,y,z\})\wedge CPath(t,t_{e},x,y,z))$ %
\begin{lyxgreyedout}
\{All of $x,y,z$ on $path_{t,t_{e}}(t,w)$\}%
\end{lyxgreyedout}
\\
$\vee(CPath(v,v_{e},v,s,\{x\})\wedge CPath(t,t_{e},t,w,\{y,z\})\wedge CPath(t,t_{e},t,y,z))$
\begin{lyxgreyedout}
\{$x$ on $path_{v,v_{e}}(v,s)$ and $y,z$ on $path_{t,t_{e}}(t,w)$\}%
\end{lyxgreyedout}
\\
$\vee(CPath(v,v_{e},v,s,\{x,z\})\wedge CPath(v,v_{e},v,z,x)\wedge CPath(t,t_{e},t,w,y))$
\begin{lyxgreyedout}
\{$x,z$ on $path_{v,v_{e}}(v,s)$ and $y$ on $path_{t,t_{e}}(t,w)$\}%
\end{lyxgreyedout}
\\

\end{itemize}
(Back to Section \ref{emp-cbfs-delete})

\end{appendices}

\end{document}